\documentclass[a4paper,11pt]{article}
\usepackage{geometry}
\usepackage[T1]{fontenc} % load this, see https://tex.stackexchange.com/questions/664/why-should-i-use-usepackaget1fontenc
\usepackage[utf8]{inputenc}
\usepackage[english]{babel}
\usepackage{amssymb}
\usepackage{amsmath}
\usepackage{amsthm}
\usepackage{graphicx}
\graphicspath{{figures}}
\usepackage{rotating} % for the sideways table

\usepackage[round]{natbib}
\usepackage{hyperref}
\hypersetup{
 unicode=true,                  % non-Latin characters in Acrobat’s bookmarks
 pdftoolbar=true,               % show Acrobat’s toolbar?
 pdfmenubar=true,               % show Acrobat’s menu?
 pdffitwindow=false,            % window fit to page when opened
 pdfstartview={FitH},           % fits the width of the page to the window
 pdftitle={Noncommuting common causes revisited},    % title
 pdfauthor={Gábor Hofer-Szabó, Szilárd Szalay},        % author
 pdfsubject={manuscript on common cause systems},      % subject of the document
 pdfcreator={pdflatex},         % creator of the document
 pdfproducer={overleaf},        % producer of the document
 pdfkeywords={common cause systems} {quantum probability} {quantum correlations} {Reichenbach's common cause principle}, % list of keywords
 pdfnewwindow=true,             % links in new window
 colorlinks=true,               % false: boxed links; true: colored links
 linktoc=page,                  % defines which part of an entry in the table of contents is made into a link
 linkcolor=blue,                % color of internal links
 citecolor=blue,                % color of links to bibliography
 filecolor=magenta,             % color of file links
 urlcolor=blue                  % color of external links
}

\newcommand{\field}[1]{\mathbb{#1}}
\newcommand{\Id}{I}
\newcommand{\Cc}[1]{{#1}^*}
\newcommand{\ee}{\mathrm{e}}

\DeclareMathOperator{\Tr}{Tr}
\DeclareMathOperator{\diag}{diag}
\DeclareMathOperator{\Lin}{Lin}
\DeclareMathOperator{\Ran}{Ran}
\DeclareMathOperator{\Ker}{Ker}

\newcommand{\cket}[1]{\vert #1 \rangle}

\newcommand{\bra}[1]{\langle #1 \vert}

\newcommand{\proj}[1]{\cket{#1}\bra{#1}}
\newcommand{\skalp}[2]{\langle #1 \vert #2 \rangle}

\providecommand{\abs}[1]{\lvert#1\rvert}
\providecommand{\bigabs}[1]{{\bigl\lvert#1\bigr\rvert}}

\providecommand{\norm}[1]{\lVert#1\rVert}

\newcommand{\set}[1]{\{ #1 \}}

\newcommand{\sset}[2]{\set{ #1 \;\vert\; #2 }}

\theoremstyle{definition}
\newtheorem{prop}{Proposition}

\newtheorem{lem}[prop]{Lemma}
\newtheorem{cor}[prop]{Corollary}

\newcommand{\tlcell}[2]{\begin{tabular}{@{}c@{}}#1\\#2\end{tabular}}

\title{\textbf{Noncommuting common causes revisited}}
\author{
\textit{Gábor Hofer-Szabó}\thanks{HUN-REN Research Center for the Humanities, Budapest, email: szabo.gabor@btk.mta.hu} \\
\textit{Szilárd Szalay}\thanks{HUN-REN Wigner Research Centre for Physics, Budapest, email: szalay.szilard@wigner.hu}
}
\date{November 8, 2024}

\begin{document}
\maketitle

\begin{abstract}
In this paper, we revisit the concept of noncommuting common causes; refute two objections raised against them, the triviality objection and the lack of causal explanatory force; and explore how their existence modifies the EPR argument. More specifically, we show that 1) product states screening off all quantum correlations do not  compromise noncommuting common causal explanations; 2) noncommuting common causes can satisfy the law of total probability; 3) perfect correlations can have indeterministic noncommuting common causes; and, as a combination of the above claims, 4) perfect correlations can have noncommuting common causes which are both nontrivial and satisfy the law of total probability. 
\vspace{0.1in}

\end{abstract}

% ******************************************************************************
% ******************************************************************************
\section{Introduction}
\label{sec:intro}

Reichenbach's Common Cause Principle \citep{Hofer-Szabo-2013c} is the claim that if there is a correlation between two events and there is no direct causal connection between the correlating events, then there exists a common cause of the correlation. \citet{Reichenbach-1956} characterized this common cause in classical probabilistic terms as a screener-off that is a set of exhaustive and mutually exclusive events which turn the correlation into probabilistic independence when conditioned on them. This classical concept of the common cause has been generalized to noncommutative event spaces in a series of papers by \citet{Hofer-Szabo-2011,Hofer-Szabo-2012,Hofer-Szabo-2013a,Hofer-Szabo-2013b}.  In the present paper, we revisit  noncommuting common causes, address two objections raised against them, and investigate how they relate to the EPR argument. Before formulating the main results of the paper, however, it is worth providing a little background to the  noncommuting common causes. 

\newpage
The notion of the classical common cause has become a central concept in the foundations of quantum theory as a causal explanation for quantum correlations between spacelike separated events \citep{Butterfield-1989}. Since Reichenbach's original definition was purely probabilistic, in order to localize the common causes in spacetime, one had to formulate the Common Cause Principle also in local physical theories \citep{Haag-1992}. The question whether the Common Cause Principle is valid in algebraic quantum field theory (AQFT) was first posed by \citet{Redei-1997,Redei-1998}. As a positive answer, \citet{Redei-2002,Redei-2007a} have shown that the Principle holds in AQFT with locally \emph{infinite} degrees of freedom and the common cause can be localized in the union of the causal pasts of the correlating events. Since AQFT with locally infinite degrees of freedom applies special (type III) von Neumann algebras, \citet{Hofer-Szabo-2011} asked the question whether the result of Rédei and Summers also holds in local quantum theories with locally \emph{finite} degrees of freedom applying simple matrix algebras. It turned out the Common Cause Principle does \emph{not} hold in quantum Ising model, that is, the status of the Principle sensitively depends on the type of the algebra. 

This was the point when Hofer-Szabó and Vecsernyés came up with the following idea. What if we give up the restriction that the common cause should commute with the correlating events? What if we allow for noncommuting common causes in the causal explanation of quantum correlations? The formalism of quantum theory with noncommutative algebras provide ample space to look for noncommuting common causes and the screening-off property can be readily generalized for quantum probability spaces. 

This move had plenty of beneficial consequences. Here it is enough to mention only the two most important results. The first benefit of allowing for noncommuting common causes was that the Common Cause Principle could be regained in local physical theories with locally finite degrees of freedom. \citet{Hofer-Szabo-2013a} have shown that in the quantum Ising model one can always find for every correlation a common cause that does not commute with the correlating events. The second, probably more important result published by \citet{Hofer-Szabo-2012,Hofer-Szabo-2013b} was that for a set of spatially separated correlating events (projections) maximally violating the Clauser-Horne inequality there could be given a common causal explanation if commutativity is abandoned between the common cause and the correlating events. Moreover, the noncommuting common cause was local and supported in the common past of the correlating events. In short, by abandoning commutativity one could give a local common causal explanation for the EPR-Bell scenario.

The concept of a noncommuting common cause which formed the basis for both of the above results, however, was not without problems and provoked sharp criticism. The general aversion towards noncommuting common causes rooted in two facts. First, using noncommuting common causes naively to address the basic problems of the causal explanation in quantum theory in a formal way is no use whatsoever, if it is not underpinned by a viable ontology on which the causal theory can be based. Noncommutative algebras, however, lack such a clear ontology and possess only an operational interpretation. Second, the construction of a local common causal explanation for the EPR-Bell scenario seemed to threaten the moral of Bell's theorems, namely that there is no local hidden variable model for quantum theory. 

The more concrete criticisms revolved especially around two points. The first objection raised by \citet{Cavalcanti-2014} was that noncommuting common causes are \emph{trivial} in the sense that any product state can be posited as a common cause for any correlations. The second, more serious objection raised both by \citet{Cavalcanti-2014} and also by \citet{Lazarovici-2014} was that noncommuting common causes \emph{lack of causal explanatory force} because they violate the law of total probability. 

To see this second objection somewhat better, it will be useful  to understand what role the  law of total probability plays in the common causal explanation for correlations. Suppose that we find that two events, ``having yellowish fingertips'' and  ``having lung cancer'' are correlating in a population and we would like to explain the correlation by assuming that the correlation is due to ``smoking'' as a common cause. In Reichenbach's terms this means that if we condition either on the smoking or on the nonsmoking subpopulation the correlation should disappear. In other words, the correlation is simply the result of mixing these two probabilistic independencies. However, it is also required from the causal explanation---even if not explicitly emphasized---that the probability or frequency of the correlating events within the whole population \textit{should be recovered} from the corresponding  probabilities of the two subpopulations. That is, the ratio of people having yellowish fingertips in the whole population should be equal to the ratio of people having yellowish fingertips in the smoking population multiplied by the ratio of  the smoking population plus the ratio of people having  yellowish fingertips in the nonsmoking population multiplied by the ratio of  the nonsmoking population. Similarly, for the ratio of people having lung cancer and for the ratio of people having both properties. This requirement is nothing but a special case of the law of total probability. The law of total probability is a theorem in classical probability theory and maybe this was the reason why this extra requirement was not explicitly mentioned in the definition of the common cause.

In the noncommutative probability theory, however, the law of total probability does not hold universally. In other words, from the simple fact that a set of events provide a screener-off for a correlation, it does not follow that the probability of the correlating event can be reproduced from the probabilities conditioned on the common causes. Since a common causal explanation is nothing but a finer description of the correlation at a more fundamental level, the violation of the law of total probability for such common causes would mean that the causal explanation fails since the common causes cannot reproduce the original correlation.

These two objections to the application of noncommuting common causes were partly tackled by \citet{Hofer-Szabo-2013b}. The noncommuting common cause provided for the EPR-Bell scenario was far from trivial and the problem due to the violation of the law of total probability was readily acknowledged. Still, the following two questions were not treated in depth:
\begin{enumerate}
\item Does the existence of trivial noncommuting common causes provide a counterargument to the application of noncommuting common causes?
\item Does noncommutativity exclude the validity of the law of total probability or do there exist noncommuting common causes which satisfy the law of total probability?
\end{enumerate}

Beyond these questions, there is a third topic which our paper aims to touch upon. It relates to a recent debate between the philosopher Tim Maudlin and the physicist Reinhard Werner on what the violation of Bell's inequality really exclude \citep{Maudlin-2014a,Werner-2014a,Maudlin_2014b,Werner-2014b}. To simplify the debate, Werner holds that Bell's theorems have \emph{two} separate assumptions, namely locality and classicality; while Maudlin defends the position that Bell's theorems have only \emph{one} assumption, namely locality, and neither classicality nor determinism is assumed in the arguments. Thus, the violation of Bell's inequality must be blamed on the violation of locality alone. 

In our paper we will focus on just one step in the debate between Maudlin and Werner, namely the derivation of determinism from locality using the EPR argument.\footnote{For a delightful paper clarifying the positions and reconstructing the arguments, see \citep{Gomori-2023}.} In the EPR argument, one starts with a perfect correlation between spacelike separated events, say, results of spin measurements. Assuming locality and the Common Cause Principle, that is assuming that there is no direct causal connection between these spacelike separated events but such a distant correlation calls for a common cause, EPR notoriously shows that a common cause accounting for a perfect correlation must be deterministic. Maudlin is right that \citet{Bell-2004} indeed used the EPR argument as a kind of zeroth step in his argument and did not assumed determinism. But the question is how much ``classicality'' is assumed in the EPR argument. Or to put it another way, if we give up ``classicality'' and allow for noncommuting common causes, does determinism still follow from locality and the Common Cause Principle? This question boils down to  the third question tackled in this paper:
\begin{enumerate}
\item[3.] Can perfect correlations have indeterministic noncommuting common causes?
\end{enumerate}

In this paper we will address the above three questions in a formal way. In the investigation we completely ignore the question of spacetime localization of the common cause and restrict ourselves to the standard Hilbert space formalism of quantum mechanics, more precisely to the case of discrete finite quantum systems. Searching for common causes with particular properties is not an easy task in general. Therefore, in this paper we do not seek to provide a general analysis of noncommuting common causes for discrete finite systems. Rather, we develop some tools for the special case of two-qubit systems (see Lemmas~\ref{lem:2x2SCRNP} and~\ref{lem:2x2SEPCOV}) to construct a host of examples for common causes with various properties. 

In the paper, we will proceed as follows. In Section~\ref{sec:prob}, we collect the basics of quantum probability theory and define the concept of the common cause system. In Sections~\ref{sec:CCSprod},~\ref{sec:CCSLTP} and~\ref{sec:CCSperfcorr} we address the above three questions separately, one in each section. In Section~\ref{sec:limitations} we investigate how common causes fare when these questions are posed at the same time. We divide each section into two subsections, one for the main concepts and one for the examples. Since we provide ample examples for many different features, for the convenience of the reader we summarize the properties of the different examples in an extended table in Appendix~\ref{appsec:2x2.summary}. In Section~\ref{sec:concl}, we conclude.

The main results of the paper are the answers to the above questions: 
\begin{enumerate}
\item  The triviality objection does not compromise the application of noncommuting common causes since atomic product quantum events screen off all quantum correlations just like classical atomic events screen off all classical correlations; still no one would say that classical common causes are meaningless just because of the existence of atomic screener-offs. Trivial common causes form only a minor part of the possible common causes for a given correlation. To support this claim, we provide in Section~\ref{sec:CCSprod} a bunch of examples for not strongly trivial (see below) common cause systems.
\item  Noncommutativity does not generally exclude the law of total probability. To show this, we exhibit in Section~\ref{sec:CCSLTP} quantum correlations with noncommuting common cause systems which satisfy the law of total probability and hence can qualify as a fine-grained description of the correlation.
\item The EPR argument assumes some ``classicality'' in the sense that deterministic common causes cannot be derived from perfect correlations if the underlying event space is not commutative. We show in Section~\ref{sec:CCSperfcorr} that perfect correlations can have  noncommuting common causes which are indeterministic.
\end{enumerate}
Finally, we show that the above three requirements---namely, nontriviality, satisfaction of the law of total probability and perfect correlation---, even together, do not rule out the existence of noncommuting common causes.
\begin{enumerate}
\item[4.] We show in Section~\ref{sec:limitations} that for perfect correlations there exist noncommutative common causes which are both nontrivial and satisfy the law of total probability. 
\end{enumerate}

% ******************************************************************************
% ******************************************************************************
\section{Noncommutative probability theory}
\label{sec:prob}

Noncommutative probability theory is the generalization of classical probability theory to noncommutative event structure. In this section, we overview some basic concepts of this theory and the definition of noncommuting common causes. For an extended treatment of quantum probability see \citep{Redei-2007b}.

% ******************************************************************************
\subsection{General theory}
\label{sec:prob.gen}

Let $\mathcal{N}$ be a von Neumann algebra, $\Id\in\mathcal{N}$ its unit,
$\mathcal{P}(\mathcal{N})$ its projection lattice (event algebra), and $\phi$ a state on $\mathcal{N}$,
being a positive linear functional, which is normalized $\phi(\Id)=1$. 
For the event $A\in\mathcal{P}(\mathcal{N})$, we have the negation $\lnot A := A^\perp=\Id-A$;
for the \emph{mutually excluding events}, represented by orthogonal projections $A,B\in\mathcal{P}(\mathcal{N})$, we have the disjunction $A\vee B = A+B$;
and for the \emph{commuting events} $A,B\in\mathcal{P}(\mathcal{N})$, we have the conjunction $A\wedge B = AB$.
The probability of an event $A\in\mathcal{P}(\mathcal{N})$ is given by $\phi(A)$,
and the conditional probability of an event $A\in\mathcal{P}(\mathcal{N})$ with respect to the condition $C\in\mathcal{P}(\mathcal{N})$ 
is given by $\phi(A|C)=\phi(CAC)/\phi(C)$ if $\phi(C)\neq0$.

Two commuting events $A,B\in\mathcal{P}(\mathcal{N})$ are said to be \emph{correlated}, if
\begin{subequations}
\label{eq:corrAB}
\begin{equation}
\label{eq:corrAB.original}
\phi(AB)\neq\phi(A)\phi(B).
\end{equation}
Writing this as $\phi(AB)\phi(\Id)\neq\phi(A\Id)\phi(\Id B)$,
and inserting $\Id=(A+A^\perp)(B+B^\perp)$, $\Id=B+B^\perp$ and $\Id=A+A^\perp$ into the second, third and fourth expectations respectively,
we end up with a more ``balanced'' expression of the correlation
\begin{equation}
\label{eq:corrAB.balanced}
\phi(AB)\phi(A^\perp B^\perp)\neq\phi(AB^\perp)\phi(A^\perp B).
\end{equation}
\end{subequations}
The difference of the two sides of~\eqref{eq:corrAB.original} or~\eqref{eq:corrAB.balanced} characterizes the correlation,
for which we again have the original,
\begin{subequations}
\label{eq:Delta}
\begin{align}
\label{eq:Delta.original}
\Delta_\phi(A,B)&:=\phi(AB)-\phi(A)\phi(B),
\intertext{and the more ``balanced'' expression}
\label{eq:Delta.balanced}
\Delta_\phi(A,B)&\phantom{:}=\phi(AB)\phi(A^\perp B^\perp)- \phi(AB^\perp)\phi(A^\perp B).
\end{align}
\end{subequations}
For this we have
\begin{equation}
\label{eq:Deltabound}
-1/4\leq \Delta_\phi(A,B) \leq 1/4,
\end{equation}
and these extrema are taken.
This can be seen by the~\eqref{eq:Delta.balanced} expression of the correlation,
with the constraint $\phi(AB)+\phi(AB^\perp)+\phi(A^\perp B)+\phi(A^\perp B^\perp)=1$,
so $\phi(AB^\perp)=\phi(A^\perp B)=0$ holds for the maximum,
then $\Delta_\phi(A,B)=\phi(AB)(1-\phi(AB))$, taking its maximum at $\phi(AB)=1/2$.
Similarly, $\phi(AB)=\phi(A^\perp B^\perp)=0$ holds for the minimum,
then $\Delta_\phi(A,B)=-\phi(AB^\perp )(1-\phi(AB^\perp ))$, taking its minimum at $\phi(AB^\perp )=1/2$.

The $\sset{C_k\in\mathcal{P}(\mathcal{N})}{k\in K}$ set of projections is a \emph{partition (of the unit)}, if $\sum_k C_k = \Id$. It consists of mutually orthogonal projections,
representing a complete (mutually exclusive and collectively exhaustive) system of events.
The \emph{conditional expectation} $E:\mathcal{N}\mapsto\mathcal{N}$ with respect to the partition $\set{C_k}_k$
is a projection given as
\begin{equation}
\label{eq:cond}
E(A):=\sum_k C_k A C_k.
\end{equation}
The \emph{conditional state}s $\phi_{C_k}$ with respect to the conditions $C_k$ of the partition $\set{C_k}_k$
are given as
\begin{equation}
\label{eq:condphi}
\forall k\in K, \phi(C_k)\neq0:\quad
\phi_{C_k}(A) := \phi(A|C_k) := \frac{(\phi\circ E)(AC_k)}{\phi(C_k)} = \frac{\phi(C_k A C_k)}{\phi(C_k)}.
\end{equation}
The \emph{conditional correlation}s of $A$ and $B$ with respect to the conditions $C_k$ are then
\begin{subequations}
\label{eq:condDelta}
\begin{align} 
\label{eq:condDelta.original}
\forall k\in K, \phi(C_k)\neq0:\quad
\Delta_\phi(A,B|C_k) 
&= \phi(AB|C_k)-\phi(A|C_k)\phi(B|C_k),
\intertext{having also a more ``balanced'' form}
\label{eq:condDelta.balanced}
\forall k\in K, \phi(C_k)\neq0:\quad
\Delta_\phi(A,B|C_k) 
&= \phi(AB|C_k)\phi(A^\perp B^\perp|C_k)- \phi(AB^\perp|C_k)\phi(A^\perp B|C_k).
\end{align}
\end{subequations}

A partition $\set{C_k}_k$
is a \emph{common cause system (CCS)} of the correlation of the events $A$ and $B$,
if its elements ``screen off the correlation'',
which means that the conditional correlation $\Delta_\phi(A,B|C_k)$ vanishes
for all condition $C_k$ of nonzero probability \citep{Hofer-Szabo-2013c}.
This gives by the definition~\eqref{eq:condDelta} the \emph{screening-off condition}
\begin{subequations}
\label{eq:screen}
\begin{align}
\label{eq:screen.original}
\forall k\in K, \phi(C_k)\neq0:&\quad 
\phi(AB|C_k)=\phi(A|C_k)\phi(B|C_k),
\intertext{or equivalently its more ``balanced'' form}
\label{eq:screen.balanced}
\forall k\in K, \phi(C_k)\neq0:&\quad 
\phi(AB|C_k)\phi(A^\perp B^\perp|C_k)=\phi(AB^\perp|C_k)\phi(A^\perp B|C_k).
\end{align}
\end{subequations}
Note that elements $C_k$ of zero probability can also be parts of a CCS,
without any condition on them.

A CCS $\set{C_k}_k$ of the events $A$ and $B$ is called \emph{deterministic}, if each of its elements determines the events $A$ and $B$
\begin{subequations}
\label{eq:determCCS}
\begin{equation}
\label{eq:determCCS.AB}
\forall k\in K, \phi(C_k)\neq0:\quad
\phi(A|C_k), \phi(B|C_k) \in \set{0,1},
\end{equation}
which holds if and only if
\begin{equation}
\label{eq:determCCS.X}
\forall k\in K, \phi(C_k)\neq0:\quad
\phi(X|C_k) \in \set{0,1},\quad X\in\set{AB,AB^\perp,A^\perp B,A^\perp B^\perp},
\end{equation}
\end{subequations}
else it is \emph{indeterministic}.
A CCS of the events $A$ and $B$, and also any partition $\set{C_k}_k$ is called \emph{commuting}, if 
\begin{subequations}
\label{eq:commutCCS}
\begin{equation}
\label{eq:commutCCS.AB}
\forall k\in K:\quad [A,C_k]=0,\;[B,C_k]=0,
\end{equation}
which holds if and only if
\begin{equation}
\label{eq:commutCCS.X}
\forall k\in K:\quad [X,C_k]=0,\quad X\in\set{AB,AB^\perp,A^\perp B,A^\perp B^\perp},
\end{equation}
\end{subequations}
else it is \emph{noncommuting}.
A CCS of the events $A$ and $B$ is called \emph{weakly commuting} with respect to the state $\phi$, if 
\begin{subequations}
\label{eq:wcommutCCS}
\begin{equation}
\label{eq:wcommutCCS.AB}
\forall k\in K, \phi(C_k)\neq0:\quad  [A,C_k]=0,\;[B,C_k]=0,
\end{equation}
which holds if and only if
\begin{equation}
\label{eq:wcommutCCS.X}
\forall k\in K, \phi(C_k)\neq0:\quad  [X,C_k]=0,\quad X\in\set{AB,AB^\perp,A^\perp B,A^\perp B^\perp},
\end{equation}
\end{subequations}
that is, $[A,C_k]\neq0$ or $[B,C_k]\neq 0$ may occur for causes $C_k$ of zero probability.
The commutativity~\eqref{eq:commutCCS} of the CCS depends only on the events $A$ and $B$,
while the weak commutativity~\eqref{eq:wcommutCCS} of the CCS depends also on the state.
(For classical systems the von Neumann algebra $\mathcal{N}$ is commutative,
which leads to a distributive event lattice $\mathcal{P}(\mathcal{N})$ and commuting common causes.)
The following Lemma turns out to be useful for weak commutativity.

\begin{lem}
\label{lem:wcomm}
For the weakly commuting partition $\set{C_k}_k$,
for the two projections
$C:=\sum_{k:\phi(C_k)\neq0} C_k$ and $D:=\sum_{k:\phi(C_k)=0} C_k$,
we have $\phi(DX)=\phi(XD)=0$ for $X\in\set{AB,AB^\perp,A^\perp B,A^\perp B^\perp}$.
\end{lem}
\begin{proof}
The projections $C$ and $D$ are orthogonal, $\Id=C+D$,
$\phi(D)=0$, and it follows from the weak commutativity~\eqref{eq:wcommutCCS} that
$[C,X]=0$ for $X\in\set{AB,AB^\perp,A^\perp B,A^\perp B^\perp}$.
We have $\phi(DX) = \phi(DXC)+\phi(DXD) = \phi(DCX)+\phi(DXD) = 0$,
where we used 
the weak commutativity~\eqref{eq:wcommutCCS} of $\set{C_k}_k$,
the orthogonality of the projections $C$ and $D$,
and that $\phi(DXD)=0$ follows from $0=\phi(D)=\phi(D^2)=\phi\bigl(D(X+X^\perp)D\bigr)=\phi(DXD)+\phi(DX^\perp D)$,
which is the sum of nonnegative terms.
\end{proof}

To appreciate the nontriviality of the CCSs constructed in the forthcoming sections, we need to refine the concept of the triviality of the CCSs.
We will call a CCS $\set{C_k}_k$ of the events $A$ and $B$ \emph{weakly trivial} if the screening-off~\eqref{eq:screen} holds for any state $\phi$ on $\mathcal{P}(\mathcal{N})$ and call it \emph{strongly trivial} if the screening-off~\eqref{eq:screen} holds for any state $\phi$ and for any two commuting events $A$ and $B$ in $\mathcal{P}(\mathcal{N})$. Thus, weakly trivial CCSs are sensitive to the algebraic but not to the probabilistic features of $\mathcal{P}(\mathcal{N})$ and strongly trivial CCSs are not sensitive either to the probabilistic or the algebraic features of $\mathcal{P}(\mathcal{N})$. A weakly trivial CCS will be a CCS for any correlation  between $A$ and $B$, irrespective of the strength and other characteristics of the correlation. A strongly trivial CCS, on the other hand, will be a CCS not just for any correlation between $A$ and $B$ but for any correlation between any other pair of events $A'$ and $B'$. Using the terminology of the literature, a strongly trivial CCS is a \textit{common common cause} \citep{Hofer-Szabo-2013c} or a \textit{joint common cause} \citep{Hofer-Szabo-2013b} of any correlation between any two events in $\mathcal{P}(\mathcal{N})$.

As simple examples, the partitions
\begin{subequations}
\label{eq:CCStrivAB}
\begin{align}
\label{eq:CCStrivAB.2}
\set{A,A^\perp},&\quad \set{B,B^\perp},\\
\label{eq:CCStrivAB.4}
\set{AB,AB^\perp,&A^\perp B,A^\perp B^\perp}
\end{align}
\end{subequations}
are weakly trivial CCSs for any correlation between $A$ and $B$,
which can easily be seen on the~\eqref{eq:screen.balanced} form of the screening-off condition.
These partitions, however, are not strongly trivial CCSs,
since there are correlations between some events $A'\neq A$ and $B'\neq B$ which they do not screen off. 
Later we will also see strongly trivial CCSs.

% ******************************************************************************
\subsection{Discrete finite systems}
\label{sec:prob.dxd}

In this paper, we restrict our attention mainly to the case of \emph{discrete finite systems}, 
that is, when $\mathcal{N}=\mathcal{B}(\mathcal{H})=\Lin(\mathcal{H})\equiv\mathcal{H}\otimes\mathcal{H}^*$
for a finite dimensional Hilbert space $\mathcal{H}$, where we also use Dirac's bra-ket formalism.
In this case there exist atomic events, which are minimal (rank one) projections in $\mathcal{P}(\mathcal{N})$,
and every event is the disjunction of atomic events.
If the CCS $\set{C_k}_k$ is given by atomic projections, then we call it \emph{atomic CCS},
then $C_k=\proj{\gamma_k}$ with the complete orthonormal set of vectors $\set{\cket{\gamma_k}\in\mathcal{H}}_k$.
Classical systems are described by commutative von Neumann algebras,
which can be embedded into the concrete von Neumann algebras $\mathcal{N}=\Lin(\mathcal{H})$
describing the general, noncommutative case,
as the subalgebra of operators being diagonal in a fixed basis.
We represent classical systems in this way.

In the discrete finite case, we consider states which can be given as 
\begin{equation}
\label{eq:dxdphi}
\phi(A)=\Tr(\rho A),   
\end{equation}
by the \emph{density operator} $\rho\in\Lin(\mathcal{H})$,
being positive semidefinite ($\rho^\dagger=\rho\geq0$) and trace-normalized ($\Tr(\rho)=1$),
where the invariant definition of the trace map 
$\Tr:\Lin(\mathcal{H})\to\field{C}$
is $\Tr(\cket{\alpha}\bra{\beta})=\skalp{\beta}{\alpha}$.
(If $\dim(\mathcal{H})>2$ then all the states can be given by density operators.)
Density operators are also called states then.
A density operator $\rho$ describes a \emph{pure state}, if it is of rank one.
The density operator representing a pure state is given as $\rho=\proj{\psi}$
where $\cket{\psi}\in\mathcal{H}$, $\norm{\psi}=1$ is called a \emph{state vector}.
Every density operator is the convex combination of pure ones.

For the discrete finite case, 
we can express the correlation~\eqref{eq:Delta} as
\begin{subequations}
\label{eq:dxdDelta}
\begin{align}
\label{eq:dxdDelta.original}
\Delta_\rho(A,B)&:=\Tr(\rho AB)-\Tr(\rho A)\Tr(\rho B),
\intertext{and as its more ``balanced'' expression}
\label{eq:dxdDelta.balanced}
\Delta_\rho(A,B)&\phantom{:}=\Tr(\rho AB)\Tr(\rho A^\perp B^\perp)- \Tr(\rho AB^\perp)\Tr(\rho A^\perp B),
\end{align}
\end{subequations}
the conditional state~\eqref{eq:condphi} as
\begin{equation}
\label{eq:dxdcondphi}
\text{if $\Tr(\rho C_k)\neq0$:}\quad
\phi(A|C_k) = \frac{\Tr(\rho C_k A C_k)}{\Tr(\rho C_k)},
\end{equation}
and also the conditional correlation~\eqref{eq:condDelta}
using the density operator and the trace~\eqref{eq:dxdphi}. 
Then the screening-off conditions~\eqref{eq:screen} take the forms
\begin{subequations}
\label{eq:dxdscreen}
\begin{align}
\label{eq:dxdscreen.original}
&\forall k\in K, \Tr(\rho C_k)\neq0:\quad
\frac{\Tr(\rho C_k AB C_k)}{\Tr(\rho C_k)}=\frac{\Tr(\rho C_k A C_k)}{\Tr(\rho C_k)}\frac{\Tr(\rho C_k B C_k)}{\Tr(\rho C_k)},
\intertext{or equivalently  }
\label{eq:dxdscreen.balanced}
\begin{split}
&\forall k\in K, \Tr(\rho C_k)\neq0:\\
&\qquad\frac{\Tr(\rho C_kABC_k)}{\Tr(\rho C_k)}\frac{\Tr(\rho C_kA^\perp B^\perp C_k)}{\Tr(\rho C_k)}=
\frac{\Tr(\rho C_kAB^\perp C_k)}{\Tr(\rho C_k)}\frac{\Tr(\rho C_kA^\perp BC_k)}{\Tr(\rho C_k)}.
\end{split}
\end{align}
\end{subequations}

\begin{prop}
\label{prop:commatomicCCSwtriv}
In discrete finite systems,
any weakly commuting~\eqref{eq:wcommutCCS} atomic partition $\set{C_k=\proj{\gamma_k}}_k$
is an atomic CCS of the events $A$ and $B$,
which is weakly trivial.
\end{prop}
\begin{proof}
This can be seen by noting that $C_k=\proj{\gamma_k}$ commutes with $A$ if and only if 
$\cket{\gamma_k}$ is in the kernel of either $A^\bot$ or $A$.
(Indeed, 
from the commutativity $A\proj{\gamma_k}=\proj{\gamma_k}A$,
we have $A\cket{\gamma_k}=\proj{\gamma_k}A\cket{\gamma_k}$, 
then $A\cket{\gamma_k}=A^n\cket{\gamma_k}=\cket{\gamma_k}(\bra{\gamma_k}A\cket{\gamma_k})^n$, which is $n$-independent,
so $\bra{\gamma_k}A\cket{\gamma_k}$ is either $1$ or $0$,
which means that $\cket{\gamma_k}\in\Ran(A)=\Ker(A^\perp)$ or $\cket{\gamma_k}\in\Ran(A^\perp)=\Ker(A)$, respectively.)
Then not only commutativity~\eqref{eq:commutCCS} but also weak commutativity~\eqref{eq:wcommutCCS} is sufficient for the~\eqref{eq:dxdscreen.balanced} form of the screening-off condition,
since both sides of the equality in~\eqref{eq:dxdscreen.balanced} vanish for $k\in K$ for which $\Tr(\rho C_k)\neq0$,
making the partition to be a CCS of the events $A$ and $B$.
This holds independently of the state $\rho$, so the CCS is weakly trivial.
\end{proof}

\begin{cor}
\label{cor:classatomicCCSstriv}
In discrete finite \emph{classical} systems (described by commutative von Neumann algebras),
any atomic partition is an atomic CCS, which is strongly trivial;
since in this case a CCS commutes with every pair of events $A$ and $B$,
so the condition of the previous Proposition holds independently of the events $A$ and $B$, which is strong triviality.
\end{cor}

Now let us consider some special cases related to atomisticity of the event lattice.

\begin{lem}
\label{lem:dxdcondphiSpec}
In discrete finite systems,
for the \emph{general partition} $\set{C_k}_k$
and the \emph{pure state} $\rho=\proj{\psi}$,
the conditional probabilities~\eqref{eq:condphi} of $A$ and $B$ are
\begin{equation}
\label{eq:condphipure}
\forall k\in K, \norm{\psi_k}^2\neq0: \quad \phi(X|C_k)=\bra{\psi_k}X\cket{\psi_k}/\norm{\psi_k}^2,
\quad \forall X\in\mathcal{P}(\mathcal{N}),
\end{equation}
with $\cket{\psi_k}:=C_k\cket{\psi}$.
For the \emph{atomic partition} $\set{C_k=\proj{\gamma_k}}_k$
and the \emph{general state} $\rho$,
the conditional probabilities~\eqref{eq:condphi} of $A$ and $B$ are
\begin{equation}
\label{eq:condphiatomic}
\forall k\in K, \bra{\gamma_k}\rho\cket{\gamma_k}\neq0: \quad \phi(X|C_k)=\bra{\gamma_k}X\cket{\gamma_k},
\quad \forall X\in\mathcal{P}(\mathcal{N}).
\end{equation}
\end{lem}

\begin{proof}
To see \eqref{eq:condphipure}, we have by the~\eqref{eq:dxdcondphi} definition of the conditional probability that
$\phi(X|C_k) 
=\Tr(\proj{\psi} C_kXC_k)/\Tr(\proj{\psi} C_k)
= \bra{\psi}C_k X C_k \cket{\psi} / \bra{\psi}C_k\cket{\psi}$
for all operators $X\in\Lin(\mathcal{H})$,
by the definition of the trace map.\\
Similarly, 
to see \eqref{eq:condphiatomic}, we have by the~\eqref{eq:dxdcondphi} definition of the conditional probability that
$\phi(X|C_k)=\Tr(\rho \proj{\gamma_k}X\proj{\gamma_k})/\Tr(\rho \proj{\gamma_k}) 
=\bra{\gamma_k}\rho\cket{\gamma_k}\bra{\gamma_k}X\cket{\gamma_k}/\bra{\gamma_k}\rho\cket{\gamma_k}
=\bra{\gamma_k}X\cket{\gamma_k}$,
by the definition of the trace map.
\end{proof}

Note that the vectors $\cket{\psi_k}$ are not normalized,
$\norm{\psi_k}^2=\Tr(\rho C_k)\leq1$, which are just the probabilities of the events $C_k$.
Note also that, although the value of the conditional probability for atomic partitions~\eqref{eq:condphiatomic} is independent of the state,
the state-dependency is still there in the condition $\bra{\gamma_k}\rho\cket{\gamma_k}\neq0$.

\begin{lem}
\label{lem:dxdscrSpec}
In discrete finite systems,
for the \emph{general partition} $\set{C_k}_k$
and the \emph{pure state} $\rho=\proj{\psi}$,
the screening-off conditions~\eqref{eq:screen} take the forms
\begin{subequations}
\label{eq:dxdscreenpure}
\begin{align}
\label{eq:dxdscreenpure.original}
\forall k\in K, \norm{\psi_k}^2\neq0:&\quad 
\frac{\bra{\psi_k}AB\cket{\psi_k}}{\norm{\psi_k}^2}
=\frac{\bra{\psi_k}A\cket{\psi_k}}{\norm{\psi_k}^2}\frac{\bra{\psi_k}B\cket{\psi_k}}{\norm{\psi_k}^2},
\intertext{or equivalently}
\label{eq:dxdscreenpure.balanced}
\forall k\in K, \norm{\psi_k}^2\neq0:&\quad 
\frac{\bra{\psi_k}AB\cket{\psi_k}}{\norm{\psi_k}^2}\frac{\bra{\psi_k}A^\perp B^\perp\cket{\psi_k}}{\norm{\psi_k}^2}
=\frac{\bra{\psi_k}AB^\perp\cket{\psi_k}}{\norm{\psi_k}^2}\frac{\bra{\psi_k}A^\perp B\cket{\psi_k}}{\norm{\psi_k}^2},
\end{align}
\end{subequations}
with $\cket{\psi_k}:=C_k\cket{\psi}$.
For the \emph{atomic partition} $\set{C_k=\proj{\gamma_k}}_k$
and the \emph{general state} $\rho$,
the screening-off conditions~\eqref{eq:screen} take the forms
\begin{subequations}
\label{eq:dxdscreenatomic}
\begin{align}
\label{eq:dxdscreenatomic.original}
\forall k\in K, \bra{\gamma_k}\rho\cket{\gamma_k}\neq0:&\quad 
\bra{\gamma_k}AB\cket{\gamma_k}
=\bra{\gamma_k}A\cket{\gamma_k}\bra{\gamma_k}B\cket{\gamma_k},
\intertext{or equivalently}
\label{eq:dxdscreenatomic.balanced}
\forall k\in K, \bra{\gamma_k}\rho\cket{\gamma_k}\neq0:&\quad 
\bra{\gamma_k}AB\cket{\gamma_k}\bra{\gamma_k}A^\perp B^\perp\cket{\gamma_k}
=\bra{\gamma_k}AB^\perp\cket{\gamma_k}\bra{\gamma_k}A^\perp B\cket{\gamma_k}.
\end{align}
\end{subequations}
\end{lem}

\begin{proof}
This is simply by using the formulae~\eqref{eq:condphipure} and~\eqref{eq:condphiatomic} in Lemma~\ref{lem:dxdcondphiSpec}
in the screening-off conditions~\eqref{eq:screen} in the two cases.
\end{proof}

Note that the second part of Lemma~\ref{lem:dxdscrSpec} tells us that
for \emph{atomic} CCSs,
the state-dependency of the screening-off conditions~\eqref{eq:dxdscreenatomic} 
is only through the conditions $\bra{\gamma_k}\rho\cket{\gamma_k}\neq0$ on the probabilities of the common causes.
As far as this holds for the state, 
the atomic CCS screens off all correlations of the events $A$ and $B$ regardless of its strength or other characterization,
since these are encoded in the state.
The point is that
if the equalities in~\eqref{eq:dxdscreenatomic} hold for all $k\in K$,
then the atomic partition $\set{C_k=\proj{\gamma_k}}_k$ is a weakly trivial (state-independent) CCS.
If the equalities in~\eqref{eq:dxdscreenatomic} fail for some $k\in K$,
then the atomic partition $\set{C_k=\proj{\gamma_k}}_k$ can still be a CCS,
if the state $\rho$ gives zero probability $\bra{\gamma_k}\rho\cket{\gamma_k}=0$ for those $k$ indices.
In this latter case the CCS is a nontrivial one, as it depends on the state.
Let us summarize this as a corollary about the triviality of the CCS.

\begin{cor}
\label{cor:atomicCCSwtriv}
In discrete finite systems,
any atomic CCS of the events $A$ and $B$ which satisfies the equalities in~\eqref{eq:dxdscreenatomic} is weakly trivial.
There exist atomic CCSs of the events $A$ and $B$, which are nontrivial,
when the equalities~\eqref{eq:dxdscreenatomic} fail for common causes $C_k$ of zero probability only.
For the sake of completeness, recall also that
not all atomic partitions are CCSs of the events $A$ and $B$ (contrary to the classical case, see Corollary~\ref{cor:classatomicCCSstriv}),
although the weakly commuting atomic partitions are always (weakly trivial) CCSs of the events $A$ and $B$ (see Proposition~\ref{prop:commatomicCCSwtriv}).
\end{cor}

A discrete finite system is called \emph{bipartite} when there is a tensor product structure given in the Hilbert space,
\begin{subequations}
\label{eq:bipartite}
\begin{equation}
\label{eq:bipartite.H}
    \mathcal{H}=\mathcal{H}_1\otimes\mathcal{H}_2,
\end{equation}
in this case we consider only the projections $A$ and $B$ (the correlation of which we deal with),
which are of the form 
\begin{equation}
\label{eq:bipartite.AB}
A=P\otimes\Id,\quad B=\Id\otimes Q,
\end{equation}
\end{subequations}
where $P$ and $Q$ are projections on $\mathcal{H}_1$ and $\mathcal{H}_2$, respectively.
A vector $\cket{\alpha}\otimes\cket{\beta}\in\mathcal{H}$,
where $\cket{\alpha}\in\mathcal{H}_1$ and $\cket{\beta}\in\mathcal{H}_2$,
is called a \emph{product element} or \emph{product vector}.
An atomic CCS $\set{C_k=\proj{\gamma_k}}_k$ is called \emph{product atomic CCS},
if it is given by product vectors $\cket{\gamma_k}=\cket{\alpha_k}\otimes\cket{\beta_k}$.
(Then and only then the projections $C_k$ are also products of atomic projections, $\proj{\gamma_k}=\proj{\alpha_k}\otimes\proj{\beta_k}$.)
A pure state $\proj{\psi}$ is called \emph{separable},
if it is given by a product state vector $\cket{\psi}=\cket{\psi_1}\otimes\cket{\psi_2}$ (which is also called separable in this case),
otherwise it is called \emph{entangled}.

% ******************************************************************************
\subsection{Two-qubit systems}
\label{sec:prob.2x2}

When constructing the examples for CCSs of particular properties in the forthcoming sections,
we restrict our attention even further inside the discrete finite bipartite case \eqref{eq:bipartite},
to the case of \emph{two-qubit systems},
that is, when $\mathcal{H}=\mathcal{H}_1\otimes\mathcal{H}_2$
with $\dim(\mathcal{H}_1)=\dim(\mathcal{H}_2)=2$.
Without the loss of generality, let us fix the computational bases $\set{\cket{0},\cket{1}}$ in $\mathcal{H}_1$ and $\mathcal{H}_2$
such that $P=\proj{0}$ and $Q=\proj{0}$ in~\eqref{eq:bipartite.AB},
so $P^\perp=\proj{1}$ and $Q^\perp=\proj{1}$.
Then
\begin{subequations}
\begin{equation}
\label{eq:2x2AB}
A=P\otimes\Id=\proj{00}+\proj{01},\qquad 
B=\Id\otimes Q=\proj{00}+\proj{10},
\end{equation}
and
\begin{equation}
\label{eq:2x2ABsubalg}
AB=\proj{00},\quad AB^\perp=\proj{01},\quad A^\perp B=\proj{10},\quad A^\perp B^\perp=\proj{11}
\end{equation}
\end{subequations}
are just the projections to the computational basis.
Then a CCS of the events $A$ and $B$ \eqref{eq:2x2AB}, and also any partition $\set{C_k}_k$ is commuting~\eqref{eq:commutCCS}
if and only if all of its elements $C_k$ are diagonal in the computational basis,
\begin{equation}
\label{eq:2x2commutCCS}
\forall k\in K:\quad C_k = \sum_{i,j\in\set{0,1}} c_k^{ij} \proj{ij},
\end{equation}
as can be seen by the expression~\eqref{eq:commutCCS.X} with~\eqref{eq:2x2ABsubalg}.

Now let us consider some special cases related to atomisticity of the event lattice.

\begin{lem}
\label{lem:2x2condphiSpec}
In two-qubit systems,
for the \emph{general partition} $\set{C_k}_k$
and the \emph{pure state} $\rho=\proj{\psi}$,
the conditional probabilities~\eqref{eq:condphi} of $A$ and $B$ are
\begin{subequations}
\label{eq:2x2phiCkABpure}
\begin{align}
\forall k\in K, \norm{\psi_k}^2\neq0: \quad 
\phi(A|C_k) &= \bigl( \abs{\skalp{00}{\psi_k}}^2 + \abs{\skalp{01}{\psi_k}}^2 \bigr)/\norm{\psi_k}^2,\\
\phi(B|C_k) &= \bigl( \abs{\skalp{00}{\psi_k}}^2 + \abs{\skalp{10}{\psi_k}}^2 \bigr)/\norm{\psi_k}^2,
\end{align}
\end{subequations}
with $\cket{\psi_k}:=C_k\cket{\psi}$.
For the \emph{atomic partition} $\set{C_k=\proj{\gamma_k}}_k$
and the \emph{general state} $\rho$,
the conditional probabilities~\eqref{eq:condphi} of $A$ and $B$ are
\begin{subequations}
\label{eq:2x2phiCkABatomic}
\begin{align}
\forall k\in K, \bra{\gamma_k}\rho\cket{\gamma_k}\neq0: \quad 
\phi(A|C_k) &= \abs{\skalp{00}{\gamma_k}}^2 + \abs{\skalp{01}{\gamma_k}}^2,\\
\phi(B|C_k) &= \abs{\skalp{00}{\gamma_k}}^2 + \abs{\skalp{10}{\gamma_k}}^2.
\end{align}
\end{subequations}
\end{lem}

\begin{proof}
These are simply
the conditional probabilities~\eqref{eq:condphipure} and~\eqref{eq:condphiatomic} in Lemma~\ref{lem:dxdcondphiSpec} in the two cases,
applied for two-qubit systems, where we have~\eqref{eq:2x2AB}.
\end{proof}

As we will see in the following two Lemmas, in two-qubit systems many properties 
(productness and screening-off of atomic CCSs of the events $A$ and $B$ \eqref{eq:2x2AB},
separability of pure states, correlation of $A$ and $B$ in pure states)
can be formulated by a common structure, given by determinants,
which turns out to be a highly useful and convenient tool in the construction of CCSs.

\begin{lem}
\label{lem:2x2SCRNP}
In two-qubit systems, for the atomic partition $\set{C_k=\proj{\gamma_k}}_k$,
$C_k$ is a product projection, that is,
the vector $\cket{\gamma_k}$ is a product vector, $\cket{\gamma_k}=\cket{\alpha_k}\otimes\cket{\beta_k}$, if and only if
\begin{subequations}
\label{eq:2x2SCRNP}
\begin{align}
\label{eq:2x2NP}
% --------
\begin{split}
\det_{ij}(\skalp{ij}{\gamma_k}) &:=
\det\begin{pmatrix}
\skalp{00}{\gamma_k} & \skalp{01}{\gamma_k} \\
\skalp{10}{\gamma_k} & \skalp{11}{\gamma_k}
\end{pmatrix}\\
&\phantom{:}\equiv
  \skalp{00}{\gamma_k} \skalp{11}{\gamma_k}
- \skalp{01}{\gamma_k} \skalp{10}{\gamma_k} = 0,
\end{split}
% --------
\intertext{and the screening-off holds for $C_k$ with respect to the events $A$ and $B$ \eqref{eq:2x2AB} if and only if}
\label{eq:2x2SCR}
\begin{split}
\bra{\gamma_k}\rho\cket{\gamma_k}\neq0:\qquad\qquad\qquad\quad&\\
\Delta_\psi(A,B|C_k) = \det_{ij}(\abs{\skalp{ij}{\gamma_k}}^2) &:=
\det\begin{pmatrix}
\abs{\skalp{00}{\gamma_k}}^2 & \abs{\skalp{01}{\gamma_k}}^2 \\
\abs{\skalp{10}{\gamma_k}}^2 & \abs{\skalp{11}{\gamma_k}}^2
\end{pmatrix}\\
&\phantom{:}\equiv
  \abs{\skalp{00}{\gamma_k}}^2 \abs{\skalp{11}{\gamma_k}}^2
- \abs{\skalp{01}{\gamma_k}}^2 \abs{\skalp{10}{\gamma_k}}^2 = 0,
\end{split}
% --------
\end{align}
\end{subequations}
in terms of the $\skalp{ij}{\gamma_k}$ expansion coefficients of the vectors $\cket{\gamma_k}$ defining the partition.
\end{lem}

\begin{proof}
For the productness~\eqref{eq:2x2NP}, note that for $2\times2$ matrices, the vanishing of the determinant
means nonmaximal rank, which is $1$ then ($0$ is excluded since $\cket{\gamma_k}$ is not the zero vector),
which means one product term in the Schmidt decomposition (or SVD) of $\cket{\gamma_k}$.
For the screening-off~\eqref{eq:2x2SCR}, we use the condition~\eqref{eq:dxdscreenatomic.balanced}
and the form of the projections~\eqref{eq:2x2ABsubalg}.
\end{proof}

Let us emphasize again (see Corollary~\ref{cor:atomicCCSwtriv}) that in this atomic case we may have weakly trivial CCSs.
If $\det_{ij}(\abs{\skalp{ij}{\gamma_k}}^2)=0$ for all $k\in K$
then the partition $\set{C_k=\proj{\gamma_k}}_k$ is a weakly trivial CCS of the events $A$ and $B$ \eqref{eq:2x2AB},
however,
if $\det_{ij}(\abs{\skalp{ij}{\gamma_k}}^2)\neq0$ for some $k\in K$
then the partition $\set{C_k=\proj{\gamma_k}}_k$ can still be a (nontrivial) CCS
for states $\rho$ giving zero probability $\bra{\gamma_k}\rho\cket{\gamma_k}=0$
for those $k$ indices.

Lemma~\ref{lem:2x2SCRNP} gives us a practical method to construct families of \emph{nonproduct} weakly trivial atomic CCSs of the events $A$ and $B$,
that is, when $\det_{ij}(\skalp{ij}{\gamma_k})\neq0$ and $\det_{ij}(|\skalp{ij}{\gamma_k}|^2)=0$ for all $k\in K$.
The method is the following.
We start with a product vector $\cket{\alpha_k}\otimes\cket{\beta_k}$, then both determinants vanish,
which can be seen by the observation $\bra{ij}\bigl(\cket{\alpha_k}\otimes\cket{\beta_k}\bigr) = \skalp{i}{\alpha_k}\skalp{j}{\beta_k}$
(the first also by~\eqref{eq:2x2NP} in Lemma~\ref{lem:2x2SCRNP}).
Next consider a nonproduct diagonal unitary $V$.
(These are in general of the form $V=e^{i\varphi_{00}}\proj{00}+e^{i\varphi_{01}}\proj{01}+e^{i\varphi_{10}}\proj{10}+e^{i\varphi_{11}}\proj{11}$,
where the phases $\varphi_{ij}$ are such that $\varphi_{00}+\varphi_{11}\neq\varphi_{01}+\varphi_{10}$.)
Suppose also that $\skalp{i}{\alpha_k},\skalp{j}{\beta_k}\neq0$ for all $i,j\in\set{0,1}$.
Then the vector $\cket{\gamma_k}:=V\cket{\alpha_k}\otimes\cket{\beta_k}$
leads to a nonproduct screening projection,
since $V$ spoils the vanishing of the first determinant but not that of the second.

\begin{lem}
\label{lem:2x2SEPCOV}
In two-qubit systems, for the pure state $\rho=\proj{\psi}$ given by the state vector $\cket{\psi}$,
we have that $\cket{\psi}=\cket{\psi_1}\otimes\cket{\psi_2}$ (separability) if and only if
\begin{subequations}
\label{eq:2x2SEPCOV}
\begin{align}
\label{eq:2x2SEP}
% --------
\begin{split}
\det_{ij}(\skalp{ij}{\psi}) &:=
\det\begin{pmatrix}
\skalp{00}{\psi} & \skalp{01}{\psi} \\
\skalp{10}{\psi} & \skalp{11}{\psi}
\end{pmatrix}\\
&\phantom{:}\equiv
  \skalp{00}{\psi} \skalp{11}{\psi}
- \skalp{01}{\psi} \skalp{10}{\psi} = 0,
\end{split}
% --------
\intertext{and the events $A$ and $B$ \eqref{eq:2x2AB} are uncorrelated, if and only if}
\label{eq:2x2COV}
\begin{split}
\Delta_\psi(A,B)=\det_{ij}(\abs{\skalp{ij}{\psi}}^2) &:=
\det\begin{pmatrix}
\abs{\skalp{00}{\psi}}^2 & \abs{\skalp{01}{\psi}}^2 \\
\abs{\skalp{10}{\psi}}^2 & \abs{\skalp{11}{\psi}}^2
\end{pmatrix}\\
&\phantom{:}\equiv
  \abs{\skalp{00}{\psi}}^2 \abs{\skalp{11}{\psi}}^2
- \abs{\skalp{01}{\psi}}^2 \abs{\skalp{10}{\psi}}^2 = 0,
\end{split}
% --------
\end{align}
\end{subequations}
in terms of the $\skalp{ij}{\psi}$ expansion coefficients of the vector $\cket{\psi}$ defining the pure state.
\end{lem}

\begin{proof}
For the separability~\eqref{eq:2x2SEP}, the proof is the same as the proof of~\eqref{eq:2x2NP}.
For the correlation~\eqref{eq:2x2COV} of $A$ and $B$, we have~\eqref{eq:Delta.balanced},
expressed by the density operator and traces~\eqref{eq:dxdphi},
and the expression of the projections~\eqref{eq:2x2ABsubalg}.
\end{proof}

We note for the interested reader that the determinant in~\eqref{eq:2x2SEP} is related to the quantity
$4\abs{ \det_{ij}(\skalp{ij}{\psi}) }^2$, which is a proper entanglement measure for pure two-qubit states,
called \emph{concurrence squared} \citep{Wootters-1998}.
Note that the determinants~\eqref{eq:2x2NP} and~\eqref{eq:2x2SEP} work in the same way
for the productness of the vectors $\cket{\gamma_k}$ and of the state vector $\cket{\psi}$.
On the other hand, the determinants~\eqref{eq:2x2SCR} and~\eqref{eq:2x2COV} work in the same way
for the conditional and for the unconditional correlations.
These follow from that
the conditional density operator $C_k\rho C_k/\Tr(C_k\rho C_k)$
is just the state corresponding to the common cause $C_k$ itself
in the special case when the state is \emph{pure} $\rho=\proj{\psi}$ and the CCS is \emph{atomic} $C_k=\proj{\gamma_k}$.

% ******************************************************************************
% ******************************************************************************
\section{Trivial and nontrivial common cause systems}
\label{sec:CCSprod}

In this section, we turn to the first question formulated in the Introduction, namely, whether the existence of trivial noncommuting common causes provides a counterargument to the application of noncommuting common causes, and argue that they don't.
In the following, first we state a proposition (Proposition~\ref{prop:prodvecCCS}) stating that product atomic CCSs are strongly trivial (\textit{cf}. Lemma~6 in \citep{Cavalcanti-2014}).
All the discussions and also the Proposition will be formulated for bipartite~\eqref{eq:bipartite} discrete finite systems (see Section~\ref{sec:prob.dxd}).
Next, we provide examples for strongly trivial, weakly trivial and nontrivial CCSs.
All the examples are formulated in two-qubit systems (see Section~\ref{sec:prob.2x2}).

% ******************************************************************************
\subsection{Triviality of common cause systems}
\label{sec:CCSprod.theory}

It might be fruitful to summarize the general situation for the discrete finite systems. 
It is already known that in the classical case the atomic partition is unique and screens off all correlations of all pairs of events,
thus these are strongly trivial (see Corollary~\ref{cor:classatomicCCSstriv}).
In the quantum case, the atomic partitions are not unique, and they can be either products or nonproducts.
The product atomic partitions always screen off all correlations of all pairs of events,
thus these are strongly trivial (see Proposition~\ref{prop:prodvecCCS} and the \textit{first example} below).
Nonproduct atomic partitions usually do not screen off, however, some do (see the \textit{second to sixth examples} below).
There are nonproduct atomic CCSs which are weakly trivial (see Corollary~\ref{cor:atomicCCSwtriv} and the \textit{second to fourth examples} below),
and there are nonproduct atomic CCSs which are nontrivial (see the \textit{fifth and sixth examples} below).   
Recall also that there are further weakly trivial CCSs, given in \eqref{eq:CCStrivAB.2}, which are nonatomic and commuting
(also \eqref{eq:CCStrivAB.4}, if the system is larger than that of two qubits).
There are also nontrivial nonatomic CCSs, commuting and noncommuting ones (see the \textit{seventh and eighth examples} below).

As for the triviality objection, we think that the existence of strongly trivial CCSs for quantum correlations detracts from the value of noncommuting common causes just as little as the existence of trivial classical common causes detracts from the value of commuting common causes for classical correlations. Shall we declare meaningless the question as to what is the common cause of having lung cancer ($A$) and having yellowish fingertips ($B$) just because the partition $\set{AB,AB^\perp,A^\perp B,A^\perp B^\perp}$, or any atomic partition screens off the correlation? Obviously not. Similarly, the question as to what is the common cause of a quantum correlation can still make sense if for such a correlation there exist strongly trivial atomic product CCSs. Atomic product CCSs form only a small part of the set of common causes for a correlation. Describing the whole set of CCSs for all correlations in general is a complicated task which is beyond the scope of this paper. We confine ourselves here to give some illuminating examples for different kinds of CCSs of different levels of triviality.

\begin{prop}
\label{prop:prodvecCCS}
In bipartite~\eqref{eq:bipartite} discrete finite systems (see Section~\ref{sec:prob.dxd}),
any correlation of any pair of events of the form~\eqref{eq:bipartite.AB} can be screened off~\eqref{eq:dxdscreen}
by any product atomic partition.
That is, any product atomic partition forms a CCS, which is strongly trivial.
\end{prop}

\begin{proof}
For the atomic partition $\set{C_k=\proj{\gamma_k}}_k$
given by the product vectors $\cket{\gamma_k}=\cket{\alpha_k}\otimes\cket{\beta_k}$,
we have for all $A$ and $B$ events of the form~\eqref{eq:bipartite.AB}
\begin{align*} 
\bra{\gamma_k}AB\cket{\gamma_k}
&= \bra{\gamma_k}P\otimes Q\cket{\gamma_k}
 = \bra{\alpha_k}P\cket{\alpha_k}\bra{\beta_k}Q\cket{\beta_k},\\
\bra{\gamma_k}A\cket{\gamma_k}
&= \bra{\gamma_k}P\otimes \Id\cket{\gamma_k}
 = \bra{\alpha_k}P\cket{\alpha_k}\bra{\beta_k}\Id\cket{\beta_k}
 = \bra{\alpha_k}P\cket{\alpha_k},\\
\bra{\gamma_k}B\cket{\gamma_k}
&= \bra{\gamma_k}\Id\otimes Q\cket{\gamma_k}
 = \bra{\alpha_k}\Id\cket{\alpha_k}\bra{\beta_k}Q\cket{\beta_k}
 = \bra{\beta_k}Q\cket{\beta_k},
 \end{align*} 
for all $k$,
so the screening-off condition~\eqref{eq:dxdscreenatomic.original} for atomic CCS holds.
\end{proof}

Note that the product atomic CCS in Proposition~\ref{prop:prodvecCCS} is noncommuting~\eqref{eq:commutCCS} in general.
It is a commuting CCS if $P$ and $Q$ are such that $[P,\proj{\alpha_k}]=0$ and $[Q,\proj{\beta_k}]=0$ for all $k$.
Note also that we have already seen in Corollary~\ref{cor:classatomicCCSstriv} that 
in discrete finite \emph{classical} systems any atomic partition is a strongly trivial CCS.
Now we have the same also as a corollary of Proposition~\ref{prop:prodvecCCS},
since in bipartite discrete finite classical systems, any atomic event is product.

% ******************************************************************************
\subsection{Examples}
\label{sec:CCSprod.examples}

Now we provide some examples in two-qubit systems (see Section~\ref{sec:prob.2x2}),
which are of different levels of triviality.
(See also Table~\ref{tab:examples} in Appendix~\ref{appsec:2x2.summary} for convenience.)
We begin with strongly trivial CCSs.

The \textit{first example}
is a one-parameter family of product atomic (then strongly trivial) CCSs of the events $A$ and $B$ \eqref{eq:2x2AB}. 
(Recall that because of the strong triviality, these are CCSs of any pair of events of the form~\eqref{eq:bipartite.AB}.)
That is, let $\set{C_k:=\proj{\gamma_k}}_k$ be
given by the product vectors $\cket{\gamma_k}=\cket{\alpha_k}\otimes\cket{\beta_k}$.
The simplest example
is given by the vectors
\begin{equation}
\label{eq:CCSclass}
\cket{\gamma_0}:=\cket{00},\quad
\cket{\gamma_1}:=\cket{01},\quad
\cket{\gamma_2}:=\cket{10},\quad
\cket{\gamma_3}:=\cket{11}.
\end{equation}
This leads to the product atomic CCS $\set{C_k=\proj{\gamma_k}}_k=\set{AB,AB^\perp,A^\perp B,A^\perp B^\perp}$,
see~\eqref{eq:2x2ABsubalg} and also~\eqref{eq:CCStrivAB.4}.
So we get back the classical case here.
This is a commuting CCS of the events $A$ and $B$ \eqref{eq:2x2AB}, see~\eqref{eq:2x2commutCCS}.
A slightly more complicated example is given by the vectors
\begin{subequations}
\label{eq:CCSGabor}
\begin{align}
\cket{\gamma_0}&:=\frac12\bigl( + \cket{00} + \cket{01} + \cket{10} + \cket{11} \bigr),\\
\cket{\gamma_1}&:=\frac12\bigl( - \cket{00} + \cket{01} - \cket{10} + \cket{11} \bigr),\\
\cket{\gamma_2}&:=\frac12\bigl( - \cket{00} - \cket{01} + \cket{10} + \cket{11} \bigr),\\
\cket{\gamma_3}&:=\frac12\bigl( + \cket{00} - \cket{01} - \cket{10} + \cket{11} \bigr).
\end{align}
It might not be seen for the first sight that also these are product vectors
\begin{equation}
\label{eq:CCSGabor.prod}
\cket{\gamma_0}=\cket{{+}{+}},\quad
\cket{\gamma_1}=\cket{{+}{-}},\quad
\cket{\gamma_2}=\cket{{-}{+}},\quad
\cket{\gamma_3}=\cket{{-}{-}},
\end{equation}
\end{subequations}
where $\cket{{\pm}}:=\frac1{\sqrt2}(\pm\cket{0}+\cket{1})$.
(So the vectors in~\eqref{eq:CCSGabor.prod} are just the $U\otimes U$ (locally unitarily) transformed vectors of~\eqref{eq:CCSclass},
where $U=\cket{{+}}\bra{0}+\cket{{-}}\bra{1}=(\cket{0}\bra{0}-\cket{0}\bra{1}+\cket{1}\bra{0}+\cket{1}\bra{1})/\sqrt{2}$
transforms $\cket{0},\cket{1}$ into $\cket{{+}},\cket{{-}}$.)
Then $\set{C_k=\proj{\gamma_k}}_k$ is a product atomic CCS of the events $A$ and $B$ \eqref{eq:2x2AB}, but now a noncommuting one, see~\eqref{eq:2x2commutCCS}.
We can also interpolate continuously between the CCSs given by~\eqref{eq:CCSclass} and~\eqref{eq:CCSGabor}
by applying the one-parameter local unitary $U_\theta\otimes U_\theta$ to the vectors in~\eqref{eq:CCSclass}
with the more general $U_\theta=c\cket{0}\bra{0}-s\cket{0}\bra{1}+s\cket{1}\bra{0}+c\cket{1}\bra{1}$ for $\theta\in\field{R}$,
where we use the shorthand notation $c:=\cos(\theta/2)$ and $s:=\sin(\theta/2)$, leading to
\begin{subequations}
\label{eq:CCSclassU}
\begin{align}
\cket{\gamma_0} &:= U_\theta\otimes U_\theta\cket{00} = (U_\theta\cket{0})\otimes(U_\theta\cket{0}) =
+ cc \cket{00} + cs \cket{01} + sc \cket{10} + ss \cket{11},\\
\cket{\gamma_1} &:= U_\theta\otimes U_\theta\cket{01} = (U_\theta\cket{0})\otimes(U_\theta\cket{1}) =
- cs \cket{00} + cc \cket{01} - ss \cket{10} + sc \cket{11},\\
\cket{\gamma_2} &:= U_\theta\otimes U_\theta\cket{10} = (U_\theta\cket{1})\otimes(U_\theta\cket{0}) =
- sc \cket{00} - ss \cket{01} + cc \cket{10} + cs \cket{11},\\
\cket{\gamma_3} &:= U_\theta\otimes U_\theta\cket{11} = (U_\theta\cket{1})\otimes(U_\theta\cket{1}) =
+ ss \cket{00} - sc \cket{01} - cs \cket{10} + cc \cket{11},
\end{align}
\end{subequations}
which gives back~\eqref{eq:CCSclass} for $\theta\in\pi\field{Z}$ and~\eqref{eq:CCSGabor} for $\theta\in\pi(\field{Z}+1/2)$ up to phases and ordering.
Then $\set{C_k=\proj{\gamma_k}}_k$ is a product atomic CCS of the events $A$ and $B$ \eqref{eq:2x2AB},
which is commuting if and only if $\theta \in \pi\field{Z}$,
otherwise it is noncommuting, see~\eqref{eq:2x2commutCCS}.
Note that by using different and more general unitaries for the two qubits in~\eqref{eq:CCSclassU}, $U_1\otimes U_2$,
all the possible two-qubit product atomic CCSs can be covered.

Now we proceed with the construction of three weakly trivial CCSs in two-qubit systems.

The \textit{second example} 
is a one-parameter family of weakly trivial nonproduct atomic CCSs of the events $A$ and $B$ \eqref{eq:2x2AB},
constructed by using the idea based upon and outlined after Lemma~\ref{lem:2x2SCRNP}.
For the roles of $\cket{\alpha_k}\otimes\cket{\beta_k}$
let us simply have the vectors $(U_\theta\cket{i})\otimes(U_\theta\cket{j})$ already constructed in~\eqref{eq:CCSclassU},
and let us transform those by the diagonal unitary $V := \proj{00}+\proj{01}+\proj{10}-\proj{11}$,
which leads to the vectors
\begin{subequations}
\label{eq:CCStwist}
\begin{align}
\cket{\gamma_0} := V \cket{\alpha_0}\otimes\cket{\beta_0} &= 
V (U_\theta\cket{0})\otimes(U_\theta\cket{0}) = + cc \cket{00} + cs \cket{01} + sc \cket{10} - ss \cket{11},\\
\cket{\gamma_1} := V \cket{\alpha_0}\otimes\cket{\beta_1} &=
V (U_\theta\cket{0})\otimes(U_\theta\cket{1}) = - cs \cket{00} + cc \cket{01} - ss \cket{10} - sc \cket{11},\\
\cket{\gamma_2} := V \cket{\alpha_1}\otimes\cket{\beta_0} &=
V (U_\theta\cket{1})\otimes(U_\theta\cket{0}) = - sc \cket{00} - ss \cket{01} + cc \cket{10} - cs \cket{11},\\
\cket{\gamma_3} := V \cket{\alpha_1}\otimes\cket{\beta_1} &=
V (U_\theta\cket{1})\otimes(U_\theta\cket{1}) = + ss \cket{00} - sc \cket{01} - cs \cket{10} - cc \cket{11},
\end{align}
\end{subequations}
which gives back~\eqref{eq:CCSclass} for $\theta\in\pi\field{Z}$ up to phases and ordering.
(Note that we could have used different unitaries for the two qubits, and also more general unitaries with complex phases.)
Then $\set{C_k=\proj{\gamma_k}}_k$ is an atomic partition by construction, since $V$ is unitary.
It consists of nonproduct projections if and only if $\theta \not\in \pi\field{Z}$,
which can be checked by~\eqref{eq:2x2NP} as
$\det_{ij}(\skalp{ij}{\gamma_k})=\pm2 c^2s^2=\pm \sin^2(\theta)/2$.
So it is a strongly trivial CCS if $\theta \in \pi\field{Z}$.
It is a weakly trivial CCS of the events $A$ and $B$ \eqref{eq:2x2AB} for all $\theta\in\field{R}$ by construction,
which can also be checked by~\eqref{eq:2x2SCR}, as $\det_{ij}(|\skalp{ij}{\gamma_k}|^2)=(c^2s^2)^2-(c^2s^2)^2=0$ for all $k\in K$.
It is a commuting CCS of the events $A$ and $B$ \eqref{eq:2x2AB} if and only if $\theta \in \pi\field{Z}$,
otherwise it is noncommuting, see~\eqref{eq:2x2commutCCS}.

The \textit{third example}
is another one-parameter family of weakly trivial nonproduct atomic CCSs of the events $A$ and $B$ \eqref{eq:2x2AB}.
For $\theta\in\field{R}$, let us have the shorthand notation $c:=\cos(\theta/2)$ and $s:=\sin(\theta/2)$,
by which let us have the one-parameter family of vectors 
\begin{subequations}
\label{eq:CCSBell}
\begin{align}
\cket{\gamma_0}&:=c\cket{0{+}}+s\cket{1{-}}=\frac1{\sqrt{2}}\bigl(+c\cket{00} +c\cket{01} -s\cket{10} +s\cket{11}\bigr),\\
\cket{\gamma_1}&:=c\cket{0{-}}+s\cket{1{+}}=\frac1{\sqrt{2}}\bigl(-c\cket{00} +c\cket{01} +s\cket{10} +s\cket{11}\bigr),\\
\cket{\gamma_2}&:=s\cket{0{-}}-c\cket{1{+}}=\frac1{\sqrt{2}}\bigl(-s\cket{00} +s\cket{01} -c\cket{10} -c\cket{11}\bigr),\\
\cket{\gamma_3}&:=s\cket{0{+}}-c\cket{1{-}}=\frac1{\sqrt{2}}\bigl(+s\cket{00} +s\cket{01} +c\cket{10} -c\cket{11}\bigr),
\end{align}
\end{subequations}
the nonproductness of which is set by the parameter $\theta$.
It is straightforward to check that $\set{C_k=\proj{\gamma_k}}_k$ is an atomic partition.
It consists of nonproduct projections if and only if $\theta\not\in\pi\field{Z}$, which can be checked by~\eqref{eq:2x2NP},
as $\det_{ij}(\skalp{ij}{\gamma_k})=\pm cs=\pm\sin(\theta)/2$.
So it is a strongly trivial CCS if $\theta \in \pi\field{Z}$.
It is a weakly trivial CCS of the events $A$ and $B$ \eqref{eq:2x2AB} for all $\theta\in\field{R}$,
which can be checked by~\eqref{eq:2x2SCR}, as $\det_{ij}(|\skalp{ij}{\gamma_k}|^2)=((cs)^2-(cs)^2)/2=0$ for all $k\in K$.
It is a noncommuting CCS of the events $A$ and $B$ \eqref{eq:2x2AB}, see~\eqref{eq:2x2commutCCS}.

The \textit{fourth example}
is another one-parameter family of weakly trivial nonproduct atomic CCSs of the events $A$ and $B$ \eqref{eq:2x2AB}.
For $\xi,\zeta\in\field{R}$, let us have the exponentially parametrized vectors
\begin{subequations}
\label{eq:CCShyper}
\begin{align}
\cket{\gamma_0} &:= N\bigl( + \ee^{+\xi/2} \cket{00} + \ee^{+\zeta/2} \cket{01} + \ee^{-\zeta/2} \cket{10} - \ee^{-\xi/2} \cket{11} \bigr),\\
\cket{\gamma_1} &:= N\bigl( + \ee^{-\xi/2} \cket{00} + \ee^{-\zeta/2} \cket{01} - \ee^{+\zeta/2} \cket{10} + \ee^{+\xi/2} \cket{11} \bigr),\\
\cket{\gamma_2} &:= N\bigl( + \ee^{+\zeta/2} \cket{00} - \ee^{+\xi/2} \cket{01} + \ee^{-\xi/2} \cket{10} + \ee^{-\zeta/2} \cket{11} \bigr),\\
\cket{\gamma_3} &:= N\bigl( - \ee^{-\zeta/2} \cket{00} + \ee^{-\xi/2} \cket{01} + \ee^{+\xi/2} \cket{10} + \ee^{+\zeta/2} \cket{11} \bigr),
\end{align}
\end{subequations}
with the normalization $N=1/\sqrt{2(\cosh(\xi) +\cosh(\zeta))}$.
It is straightforward to check that $\set{C_k=\proj{\gamma_k}}_k$ is an atomic partition.
It never consists of product projections, which can be checked by~\eqref{eq:2x2NP},
as $\det_{ij}(\skalp{ij}{\gamma_k})=\pm 2N^2\neq0$.
It is a weakly trivial CCS of the events $A$ and $B$ \eqref{eq:2x2AB} for all $\theta\in\field{R}$,
which can be checked by~\eqref{eq:2x2SCR}, as $\det_{ij}(|\skalp{ij}{\gamma_k}|^2)=0$ for all $k\in K$.
It is a noncommuting CCS of the events $A$ and $B$ \eqref{eq:2x2AB}, see~\eqref{eq:2x2commutCCS}.

Now we turn to nontrivial CCSs in two qubit systems, starting with two atomic ones.
Note that nontrivial CCSs are state-dependent, so we have to specify the state as well.
In these examples, the state is pure and given by the state vector
\begin{equation}
\label{eq:Bell0}
\cket{\psi}:=\bigl(\cket{00}+\cket{11}\bigr)/\sqrt2,
\end{equation}
which is entangled, which can be checked by~\eqref{eq:2x2SEP},
as $\det_{ij}(\skalp{ij}{\psi})=1/2$.

The \textit{fifth example}
is a one-parameter family of nontrivial (then nonproduct) atomic CCSs of the events $A$ and $B$ \eqref{eq:2x2AB}
in the system described by the state given by the state vector~\eqref{eq:Bell0}.
For $\theta\in\field{R}$, let us have the shorthand notation $c:=\cos(\theta/2)$ and $s:=\sin(\theta/2)$,
by which let us have the one-parameter family of vectors 
\begin{subequations}
\label{eq:CCSntrat}
\begin{align}
    \cket{\gamma_0}&:=\cket{00},\\
    \cket{\gamma_1}&:=c\cket{01}+s\cket{10},\\
    \cket{\gamma_2}&:=-s\cket{01}+c\cket{10},\\
    \cket{\gamma_3}&:=\cket{11},
\end{align}
\end{subequations}
which give back~\eqref{eq:CCSclass} for $\theta\in\pi\field{Z}$ up to phases and ordering.
It is straightforward to check that $\set{C_k=\proj{\gamma_k}}_k$ is an atomic partition.
It contains nonproduct projections if and only if $\theta\not\in\pi\field{Z}$, which can be checked by~\eqref{eq:2x2NP},
as $\det_{ij}(\skalp{ij}{\gamma_k})=\pm cs=\pm\sin(\theta)/2$ for $k\in\set{1,2}$.
So it is a strongly trivial CCS if $\theta \in \pi\field{Z}$.
It is a CCS of the events $A$ and $B$ \eqref{eq:2x2AB} for the state~\eqref{eq:Bell0} for all $\theta\in\field{R}$,
which can be checked by~\eqref{eq:2x2SCR},
as $\det_{ij}(\abs{\skalp{ij}{\gamma_0}}^2)=\det_{ij}(\abs{\skalp{ij}{\gamma_3}}^2)=0$,
and although $\det_{ij}(\abs{\skalp{ij}{\gamma_1}}^2)=\det_{ij}(\abs{\skalp{ij}{\gamma_2}}^2)=-c^2s^2\neq0$,
$C_1$ and $C_2$ are of zero probability, 
$\abs{\skalp{\psi}{\gamma_1}}^2=\abs{\skalp{\psi}{\gamma_1}}^2=0$.
So it is a nontrivial CCS if $\theta \not\in \pi\field{Z}$.
It is a commuting CCS of the events $A$ and $B$ \eqref{eq:2x2AB} if and only if $\theta\in\pi\field{Z}$,
otherwise it is weakly commuting~\eqref{eq:wcommutCCS}, see~\eqref{eq:2x2commutCCS},
since the commutativity~\eqref{eq:commutCCS} fails only for $C_1$ and $C_2$, which are of zero probability.

The \textit{sixth example}
is another one-parameter family of nontrivial (then nonproduct) atomic CCSs of the events $A$ and $B$ \eqref{eq:2x2AB}
in the system described by the state given by the state vector~\eqref{eq:Bell0}.
For $\theta\in\field{R}$, let us have the shorthand notation $c:=\cos(\theta/2)$ and $s:=\sin(\theta/2)$,
by which let us have the one-parameter family of vectors 
\begin{subequations}
\label{eq:CCSntratU}
\begin{align}
    \cket{\gamma_0}&:=\cket{{+}{+}},\\
    \cket{\gamma_1}&:=c\cket{{+}{-}}+s\cket{{-}{+}},\\
    \cket{\gamma_2}&:=-s\cket{{+}{-}}+c\cket{{-}{+}},\\
    \cket{\gamma_3}&:=\cket{{-}{-}},
\end{align}
\end{subequations}
which are the
$U\otimes U$ (locally unitarily) transformed vectors of \eqref{eq:CCSntrat},
giving back~\eqref{eq:CCSGabor} for $\theta\in\pi\field{Z}$ up to phases and ordering.
Note also that the state vector $\cket{\psi}$ given in~\eqref{eq:Bell0} is invariant to local unitary transformations of the form $U\otimes U$.
It is straightforward to check that $\set{C_k=\proj{\gamma_k}}_k$ is an atomic partition,
this is invariant to local unitary transformations.
It contains nonproduct projections if and only if $\theta\not\in\pi\field{Z}$, 
since productness is invariant to local unitary transformations,
although it may also be illustrative to check this directly by~\eqref{eq:2x2NP},
as $\det_{ij}(\skalp{ij}{\gamma_k})=\pm cs=\pm\sin(\theta)/2$ for $k\in\set{1,2}$.
So it is a strongly trivial CCS if $\theta \in \pi\field{Z}$.
It is a CCS of the events $A$ and $B$ \eqref{eq:2x2AB} for the state~\eqref{eq:Bell0} for all $\theta\in\field{R}$,
which can be checked by~\eqref{eq:2x2SCR},
as $\det_{ij}(\abs{\skalp{ij}{\gamma_0}}^2)=\det_{ij}(\abs{\skalp{ij}{\gamma_3}}^2)=0$,
and although $\det_{ij}(\abs{\skalp{ij}{\gamma_1}}^2)=\det_{ij}(\abs{\skalp{ij}{\gamma_2}}^2)=-c^2s^2/4\neq0$,
$C_1$ and $C_2$ are of zero probability, 
$\abs{\skalp{\psi}{\gamma_1}}^2=\abs{\skalp{\psi}{\gamma_1}}^2=0$.
So it is a nontrivial CCS if $\theta \not\in \pi\field{Z}$.
(Note that neither the correlation in~\eqref{eq:2x2COV} nor the conditional correlation in~\eqref{eq:2x2SCR}
is invariant to the local unitary transformations of the form $U\otimes U$ used here,
since the events $A$ and $B$ are fixed in those formulas.)
It is a noncommuting CCS of the events $A$ and $B$ \eqref{eq:2x2AB}, see~\eqref{eq:2x2commutCCS}.

Finally we construct two nontrivial CCSs in two-qubit systems which are nonatomic.
Note that nontrivial CCSs are state-dependent, so we have to specify the state as well.

The \textit{seventh example}
is a one-parameter family of nontrivial nonatomic CCSs of the events $A$ and $B$ \eqref{eq:2x2AB}
in the system described by the state given by the state vector~\eqref{eq:Bell0}.
Let us have
\begin{equation}
\label{eq:CCS22ntrat}
C_+:=\proj{\gamma_0}+\proj{\gamma_1},\qquad C_-:=\proj{\gamma_2}+\proj{\gamma_3}
\end{equation}
with the $\cket{\gamma_k}$ vectors given in~\eqref{eq:CCSntrat},
covering also the weakly trivial CCSs given in~\eqref{eq:CCStrivAB.2} for $\theta\in\pi2\field{Z}$ and for $\theta\in\pi(2\field{Z}+1)$.
So $\set{C_+,C_-}$ is a weakly trivial CCS if $\theta \in \pi\field{Z}$.
It is a rank-two CCS of the events $A$ and $B$ \eqref{eq:2x2AB} for the state~\eqref{eq:Bell0} for all $\theta\in\field{R}$,
which can be checked by~\eqref{eq:dxdscreenpure.balanced},
as $\cket{\psi_+}=C_+\cket{\psi}=\cket{\gamma_0}/\sqrt2=\cket{00}/\sqrt2$ and
$\cket{\psi_-}=C_-\cket{\psi}=\cket{\gamma_3}/\sqrt2=\cket{11}/\sqrt2$,
since
$\skalp{\gamma_1}{\psi}=\skalp{\gamma_2}{\psi}=0$.
So it is a nontrivial CCS if $\theta \not\in \pi\field{Z}$.
It is a commuting CCS of the events $A$ and $B$ \eqref{eq:2x2AB} if and only if $\theta \in \pi\field{Z}$,
otherwise it is noncommuting, see~\eqref{eq:2x2commutCCS}.

The \textit{eighth example}
is another one-parameter family of nontrivial nonatomic CCSs of the events $A$ and $B$ \eqref{eq:2x2AB},
in the system described by the state given by the state vector~\eqref{eq:Bell0}.
Let us have
\begin{equation}
\label{eq:CCS22ntratU}
C_+:=\proj{\gamma_0}+\proj{\gamma_1},\qquad C_-:=\proj{\gamma_2}+\proj{\gamma_3}
\end{equation}
with the $\cket{\gamma_k}$ vectors given in~\eqref{eq:CCSntratU}.
It is a rank-two CCS of the events $A$ and $B$ \eqref{eq:2x2AB} for the state~\eqref{eq:Bell0},
which can be checked by~\eqref{eq:dxdscreenpure.balanced},
as $\cket{\psi_+}=C_+\cket{\psi}=\cket{\gamma_0}/\sqrt2=\cket{{+}{+}}/\sqrt2$ and
$\cket{\psi_-}=C_-\cket{\psi}=\cket{\gamma_3}/\sqrt2=\cket{{-}{-}}/\sqrt2$,
since
$\skalp{\gamma_1}{\psi}=\skalp{\gamma_2}{\psi}=0$.
(This latter is not so easy to see directly as in the previous example, however,
it follows simply from that $\cket{\psi}$ is invariant to the transformation $U\otimes U$,
by which the vectors in~\eqref{eq:CCSntratU} used here were constructed from the vectors in~\eqref{eq:CCSntrat}.)
Then
$A       B      \cket{\psi_\pm}=\cket{00}/(2\sqrt{2})$,
$A       B^\perp\cket{\psi_\pm}=\pm\cket{01}/(2\sqrt{2})$,
$A^\perp B      \cket{\psi_\pm}=\pm\cket{10}/(2\sqrt{2})$,
$A^\perp B^\perp\cket{\psi_\pm}=\cket{11}/(2\sqrt{2})$,
then all the four terms in the screening-off condition~\eqref{eq:dxdscreenpure.balanced} (the norm-squared of these vectors) are the same.
So $\set{C_+,C_-}$ is a nontrivial CCS.
It is a noncommuting CCS of the events $A$ and $B$ \eqref{eq:2x2AB}, see~\eqref{eq:2x2commutCCS}.

% ******************************************************************************
% ******************************************************************************
\section{Noncommuting common causes and the law of total probability}
\label{sec:CCSLTP}

In this section, we turn to the second question formulated in the Introduction: Does noncommutativity exclude the validity of the law of total probability or do there exist noncommuting common causes which satisfy the law of total probability? In the following, first we define the \emph{law of total probability} in the general theory (see Section~\ref{sec:prob.gen}), in discrete finite systems (see Section~\ref{sec:prob.dxd}) and in two-qubit systems (see Section~\ref{sec:prob.2x2}) in turn. Next, we provide examples to show that being a CCS and satisfying the law of total probability are independent notions. As an answer to the above question, we will show a set of noncommuting CCSs satisfying the law of total probability. All the examples are formulated in two-qubit systems (see Section~\ref{sec:prob.2x2}).

% ******************************************************************************
\subsection{Law of total probability}
\label{sec:CCSLTP.theory}

\textit{In the general theory} (see Section~\ref{sec:prob.gen}),
the state $\phi$ and the $\set{C_k}_k$ partition together 
\emph{obey the law of total probability with respect to the commuting events $A,B\in\mathcal{P}(\mathcal{N})$}, if
the average of the conditional states~\eqref{eq:condphi} weighted with the probabilities of the conditions
recovers the probabilities of the events generated by $A$ and $B$,
that is,
\begin{subequations}
\label{eq:LTP}
\begin{align}
\label{eq:LTP.gen}
\sum_{k\in K:\phi(C_k)\neq0} \phi(C_k) \phi(X|C_k) = \phi(X)&\quad
\text{for $X\in\set{AB,AB^\perp,A^\perp B,A^\perp B^\perp}$},
\intertext{or, equivalently,}
\label{eq:LTP.simple}
\phi\Bigl(\sum_{k\in K} C_k X C_k\Bigr) = \phi(X)&\quad
\text{for $X\in\set{AB,AB^\perp,A^\perp B,A^\perp B^\perp}$},
\end{align}
\end{subequations}
that is, $\phi\circ E$ agrees with $\phi$ on the subalgebra generated by $A$ and $B$, see~\eqref{eq:cond}. 
(The equivalence above is because $\phi(C_kXC_k)=0$ follows from $\phi(C_k)=0$,
as $0=\phi(C_k)=\phi(C_k\Id C_k)=\phi(C_kXC_k)+\phi(C_kX^\perp C_k)$, being the sum of nonnegative terms.)
This always holds in the classical case (commutative von Neumann algebras),
not only for the events $X\in\set{AB,AB^\perp,A^\perp B,A^\perp B^\perp}$, but for all $X\in\mathcal{P}(\mathcal{N})$,
since $\sum_k C_k X C_k = \sum_k C_k^2 X = \sum_k C_k X = \Id X = X$.
It is clear that in the quantum case 
the law of total probability cannot hold in such a general form, that is, for all $X\in\mathcal{P}(\mathcal{N})$. 
We will see, however, that it still can hold for the events generated by $A$ and $B$, as given in~\eqref{eq:LTP},
even for noncommuting CCSs.
On the other hand, since commutativity of the CCS~\eqref{eq:commutCCS} is given with respect to $A$ and $B$,
we immediately have the following.
\begin{prop}
\label{prop:CommLTP}
Any CCS $\set{C_k}_k$ weakly commuting~\eqref{eq:wcommutCCS} with respect to the state $\phi$ obeys the law of total probability \eqref{eq:LTP}.
\end{prop}

\begin{proof}
Let us have the two projections
$C:=\sum_{k:\phi(C_k)\neq0} C_k$ and $D:=\sum_{k:\phi(C_k)=0} C_k$.
We have then $\phi(X)=\phi(CX)+\phi(DX)$, where the second term is zero by Lemma~\ref{lem:wcomm}.
Then the law of total probability~\eqref{eq:LTP.gen} follows as
$\phi(X)=\phi(CX)=\phi\bigl(\sum_{k:\phi(C_k)\neq0} C_kX\bigr)=\sum_{k:\phi(C_k)\neq0}\phi(C_k^2X)=\sum_{k:\phi(C_k)\neq0}\phi(C_kXC_k)
=\sum_{k:\phi(C_k)\neq0}\phi(C_k)\phi(X|C_k)$,
where the weak commutativity~\eqref{eq:wcommutCCS} of $\set{C_k}_k$ was used.
\end{proof}

\textit{In discrete finite systems} (see Section~\ref{sec:prob.dxd}),
the law of total probability~\eqref{eq:LTP} takes the form
\begin{equation}
\label{eq:LTPrho}
\Tr(E(\rho) X)= \Tr(\rho X)\quad\text{for $X\in\set{AB,AB^\perp,A^\perp B,A^\perp B^\perp}$},
\end{equation}
by~\eqref{eq:cond} and~\eqref{eq:dxdphi}, since $\Tr(\rho E(X)) \equiv \Tr(E(\rho) X)$.

Now let us consider a special case related to atomisticity of the event lattice.

\begin{lem}
\label{lem:dxdLTPatomicgen}
In discrete finite systems (see Section~\ref{sec:prob.dxd}),
for an atomic CCS $\set{C_k=\proj{\gamma_k}}_k$,
the law of total probability~\eqref{eq:LTPrho} takes the form
\begin{equation}
\label{eq:LTPrhovec}
\sum_{k\in K} q_k \bra{\gamma_k} X \cket{\gamma_k} = \Tr(\rho X)
\quad\text{for $X\in\set{AB,AB^\perp,A^\perp B,A^\perp B^\perp}$},
\end{equation}
where $q_k:=\bra{\gamma_k}\rho\cket{\gamma_k}$, being the probability of $C_k=\proj{\gamma_k}_k$.
\end{lem}

\begin{proof}
To see this, we have for the form~\eqref{eq:LTPrho} that
$\Tr(E(\rho)X)
=\Tr(\sum_k C_k\rho C_k X)
=\Tr(\sum_k \proj{\gamma_k}\rho\proj{\gamma_k} X)
=\sum_k \bra{\gamma_k}\rho\cket{\gamma_k}\Tr(\proj{\gamma_k} X)$,
which equals to the left-hand side of~\eqref{eq:LTPrhovec}.
\end{proof}

\textit{In two-qubit systems} (see Section~\ref{sec:prob.2x2}),
the law of total probability~\eqref{eq:LTPrho} takes the form
\begin{equation}
\label{eq:LTPrho22}
\diag(E(\rho)) = \diag(\rho),
\end{equation}
where $\diag(Y)=(\bra{00}Y\cket{00},\bra{01}Y\cket{01},\bra{10}Y\cket{10},\bra{11}Y\cket{11})$
is the tuple of diagonal elements of the matrix of the operator,
in the computational basis, see~\eqref{eq:2x2ABsubalg}.

Now let us consider some special cases related to atomisticity of the event lattice.

\begin{lem}
\label{lem:2x2LTPSpec}
In two-qubit systems (see Section~\ref{sec:prob.2x2}),
for the \emph{general CCS} $\set{C_k}_k$ of the events $A$ and $B$
and the \emph{pure state} $\rho=\proj{\psi}$,
the law of total probability for $A$ and $B$ takes the form
\begin{equation}
\label{eq:2x2LTPpure}
\sum_{k\in K} \abs{\skalp{ij}{\psi_k}}^2 =\abs{\skalp{ij}{\psi}}^2 \qquad\forall i,j\in\set{0,1},
\end{equation}
with $\cket{\psi_k}=C_k\cket{\psi}$.
For the \emph{atomic CCS} $\set{C_k=\proj{\gamma_k}}_k$ of the events $A$ and $B$
and the \emph{general state} $\rho$,
the law of total probability for $A$ and $B$ takes the form
\begin{equation}
\label{eq:2x2LTPatomic}
\sum_{k\in K} q_k \abs{\skalp{ij}{\gamma_k}}^2 =\bra{ij}\rho\cket{ij} \qquad\forall i,j\in\set{0,1},
\end{equation}
where $q_k := \bra{\gamma_k}\rho\cket{\gamma_k}$, being the probability of $C_k=\proj{\gamma_k}_k$.
For the \emph{atomic CCS} $\set{C_k=\proj{\gamma_k}}_k$ of the events $A$ and $B$
and the \emph{pure state} $\rho=\proj{\psi}$,
the law of total probability for $A$ and $B$
takes the form
\begin{equation}
\label{eq:2x2LTP}
\sum_{k\in K} q_k \abs{\skalp{ij}{\gamma_k}}^2 =\abs{\skalp{ij}{\psi}}^2 \qquad\forall i,j\in\set{0,1},
\end{equation}
where $q_k := \abs{\skalp{\psi}{\gamma_k}}^2$, being the probability of $C_k=\proj{\gamma_k}_k$.
\end{lem}

\begin{proof}
These are simply~\eqref{eq:LTPrho} in the first, and~\eqref{eq:LTPrhovec} in the second and third cases,
applied for two-qubit systems, where we have~\eqref{eq:2x2ABsubalg}.
\end{proof}

It must be obvious from the general definition of the common cause~\eqref{eq:screen} and the law of total probability~\eqref{eq:LTP} that the two concepts are logically independent. We know that in the (weakly) commutative case the law of total probability holds for any CCS (Proposition~\ref{prop:CommLTP}), whereas in the noncommutative case this is typically not the case. For example, the CCS given by \citet{Hofer-Szabo-2012,Hofer-Szabo-2013b} for the EPR scenario did violate the law of total probability; which was the basis of the criticism of \citet{Cavalcanti-2014} and \citet{Lazarovici-2014} against the concept of the noncommuting common cause. But since the concept of noncommuting common cause is logically independent of the law of total probability, this counter-example does not mean that \emph{any} noncommuting common cause a priori violates the the law of total probability. And indeed, as we show below, there exist noncommuting common causes which satisfy the law of total probability. Thus, the satisfaction of the law of total probability can count as an extra filter to search for common causal explanations for correlations.

% ******************************************************************************
\subsection{Examples}
\label{sec:CCSLTP.examples}

Now we provide some examples in two-qubit systems (see Section~\ref{sec:prob.2x2})
to show that being a CCS and satisfying the law of total probability are independent notions.  
Namely, we provide examples in turn for 
an atomic partition which is not a CCS;
two partitions which are trivial CCSs; one product atomic (so strongly trivial) and one weakly trivial nonproduct atomic;
and then two partitions which are nontrivial CCSs, one atomic and one nonatomic---all the five satisfying the law of total probability for respective states.
(See also Table~\ref{tab:examples} in Appendix~\ref{appsec:2x2.summary} for convenience.)
Recall that
the atomic CCSs~\eqref{eq:CCSclass}-\eqref{eq:CCShyper} of Section~\ref{sec:CCSprod} are weakly trivial: they screen off \emph{any} correlation of the events $A$ and $B$ \eqref{eq:2x2AB} regardless the state which describes the correlation. Atomic CCSs, however, do not necessarily satisfy the law of total probability for \emph{any} state. Therefore, for the analysis of the law of total probability, in our examples we need to specify also the state (see~\eqref{eq:LTPrhovec}). 

The \textit{first example}
is a partition satisfying the law of total probability,
although \emph{not} being a CCS of the events $A$ and $B$ \eqref{eq:2x2AB}. 
Let us have the vectors
\begin{subequations}
\label{eq:CLTP}
\begin{align}
\cket{\gamma_0}&:=\frac1{\sqrt{2}}\bigl( \cket{00} +i\cket{11} \bigr),\\
\cket{\gamma_1}&:=\frac1{\sqrt{2}}\bigl(i\cket{01} + \cket{10} \bigr),\\
\cket{\gamma_2}&:=\frac1{\sqrt{2}}\bigl(i\cket{01} - \cket{10} \bigr),\\
\cket{\gamma_3}&:=\frac1{\sqrt{2}}\bigl( \cket{00} -i\cket{11} \bigr),
\end{align}
\end{subequations}
then $\set{C_k=\proj{\gamma_k}}_k$ is an atomic partition. 
It is \emph{not} a CCS of the events $A$ and $B$ \eqref{eq:2x2AB},
which can be checked by~\eqref{eq:2x2SCR},
as $\det_{ij}(|\skalp{ij}{\gamma_k}|^2)=\pm1/4$.
It consists of nonproduct projections, which can be checked by~\eqref{eq:2x2NP},
as $\det_{ij}(\skalp{ij}{\gamma_k})=\pm i/2$.
We also have that $[C_k,A]\neq0$, $[C_k,B]\neq0$, which can be checked by~\eqref{eq:2x2commutCCS}.
For the law of total probability we also have to specify the state.
Let us have the state vector 
\begin{equation}
\label{eq:CLTPstate}
\cket{\psi}:= \frac1{\sqrt{2}}\bigl( a\cket{00} + b\cket{01} + b\cket{10} + a\cket{11} \bigr)
\end{equation}
with the parameters $a,b\in\field{C}$, $\abs{a}^2+\abs{b}^2=1$.
These are entangled if and only if $b\neq\pm a$, which can be checked by~\eqref{eq:2x2SEP},
as $\det_{ij}(\skalp{ij}{\psi})=(a^2-b^2)/2$.
In this state, $A$ and $B$ are correlated if and only if $\abs{b}\neq\abs{a}$, which can be checked by~\eqref{eq:2x2COV},
as $\Delta_\psi(A,B)=\det_{ij}(|\skalp{ij}{\psi}|^2)=(\abs{a}^4-\abs{b}^4)/4 = (\abs{a}^2-\abs{b}^2)/4$.
The law of total probability~\eqref{eq:2x2LTP} holds for $\set{C_k=\proj{\gamma_k}}_k$ given by~\eqref{eq:CLTP}
in the state $\rho=\proj{\psi}$ given by~\eqref{eq:CLTPstate},
which can be checked by~\eqref{eq:2x2LTP},
as, denoting $\omega:=(1-i)/\sqrt{2}$,
$q_k=\abs{\skalp{\psi}{\gamma_k}}^2 = (\abs{a\omega}^2,\abs{b\Cc{\omega}}^2,\abs{b\omega}^2,\abs{a\Cc{\omega}}^2)_k/2
=(\abs{a}^2,\abs{b}^2,\abs{b}^2,\abs{a}^2)_k/2$;
then $\bigl(q_0|\skalp{ij}{\gamma_0}|^2+q_3|\skalp{ij}{\gamma_3}|^2\bigr) +
\bigl(q_1|\skalp{ij}{\gamma_1}|^2+q_2|\skalp{ij}{\gamma_2}|^2\bigr) =
\abs{a}^2 (1,0,0,1)_{ij}/2 + \abs{b}^2 (0,1,1,0)_{ij}/2$,
which is just $|\skalp{ij}{\psi}|^2=(\abs{a}^2,\abs{b}^2,\abs{b}^2,\abs{a}^2)_{ij}/2$.

The \textit{second example}
is a one-parameter family of product atomic (then strongly trivial) CCSs of the events $A$ and $B$ \eqref{eq:2x2AB},
which satisfies the law of total probability.
The CCS given by~\eqref{eq:CCSclass}, being classical, obeys the law of total probability for all states $\rho$,
as can easily be seen also by direct calculation.
On the other hand, one can find pure states for which the locally transformed product CCS given by~\eqref{eq:CCSclassU}
obeys the law of total probability.
Such state vectors can be found among the ones which are of the form
\begin{equation}
\label{eq:CCSclassULTPstate}
\cket{\psi} := \frac1{\sqrt2} \bigl(a \cket{00} + b \cket{01} + b \cket{10} - a \cket{11}\bigr),
\end{equation}
with the parameters $a,b\in\field{R}$, $a^2+b^2=1$.
These are always entangled, which can be checked by~\eqref{eq:2x2SEP},
as $\det_{ij}(\skalp{ij}{\psi})=-1/2$.
In these states, the events $A$ and $B$ are correlated if and only if $\abs{b}\neq\abs{a}$, which can be checked by~\eqref{eq:2x2COV},
as $\Delta_\psi(A,B)=\det_{ij}(|\skalp{ij}{\psi}|^2)=(\abs{a}^4-\abs{b}^4)/4 = (\abs{a}^2-\abs{b}^2)/4$.
The parameters $a,b$ for which the state~\eqref{eq:CCSclassULTPstate}
and the CCS~\eqref{eq:CCSclassU} obey the law of total probability for the events $A$ and $B$
can be obtained as functions of the $c,s$ parameters of the CCS~\eqref{eq:CCSclassU}.
For the proof and the general formulas, see Appendix~\ref{appsec:2x2.CCSLTP},
here we just show a particular solution.
Let us have the vectors
\begin{subequations}
\label{eq:CCSclassUspec}
\begin{align}
\cket{\gamma_0} &:= \frac1{2\sqrt2} \bigl(   (\sqrt{2}+1) \cket{00} + \cket{01} + \cket{10} + (\sqrt{2}-1) \cket{11} \bigr),\\
\cket{\gamma_1} &:= \frac1{2\sqrt2} \bigl( - \cket{00} + (\sqrt{2}+1) \cket{01} - (\sqrt{2}-1) \cket{10} + \cket{11} \bigr),\\
\cket{\gamma_2} &:= \frac1{2\sqrt2} \bigl( - \cket{00} - (\sqrt{2}-1) \cket{01} + (\sqrt{2}+1) \cket{10} + \cket{11} \bigr),\\
\cket{\gamma_3} &:= \frac1{2\sqrt2} \bigl(   (\sqrt{2}-1) \cket{00} - \cket{01} - \cket{10} + (\sqrt{2}+1) \cket{11} \bigr),
\end{align}
\end{subequations}
then $\set{C_k=\proj{\gamma_k}}_k$ is a noncommuting product atomic CCS of the events $A$ and $B$ \eqref{eq:2x2AB},
as follows from the general properties of the CCS given by~\eqref{eq:CCSclassU},
although it may also be illustrative to check this directly by~\eqref{eq:2x2commutCCS}, \eqref{eq:2x2NP} and~\eqref{eq:2x2SCR}.
Let us also have the state vector
\begin{equation}
\label{eq:CCSclassULTPstatespec}
\cket{\psi} = \frac1{2\sqrt[4]{5}}
 \Bigl( \sqrt{\sqrt{5}+1}\cket{00} + \sqrt{\sqrt{5}-1}\cket{01} + \sqrt{\sqrt{5}-1}\cket{10} - \sqrt{\sqrt{5}+1}\cket{11}\Bigr).
\end{equation}
It is entangled, which can be checked by~\eqref{eq:2x2SEP},
as $\det_{ij}(\skalp{ij}{\psi})=-1/2$.
In this state, the events $A$ and $B$ are correlated, which can be checked by~\eqref{eq:2x2COV},
as $\Delta_\psi(A,B)=\det_{ij}(|\skalp{ij}{\psi}|^2)=1/(4\sqrt5)$.
The law of total probability~\eqref{eq:2x2LTP} holds for $\set{C_k=\proj{\gamma_k}}_k$ given by~\eqref{eq:CCSclassUspec}
in the state $\rho=\proj{\psi}$ given by~\eqref{eq:CCSclassULTPstatespec},
which follows from the general construction (see Appendix~\ref{appsec:2x2.CCSLTP}),
although it may also be illustrative to this check directly by~\eqref{eq:2x2LTP},
as $q_k=\abs{\skalp{\psi}{\gamma_k}}^2 = (\sqrt5+2,\sqrt5-2,\sqrt5-2,\sqrt5+2)_k/(4\sqrt5)$;
then $\bigl(q_0|\skalp{ij}{\gamma_0}|^2+q_3|\skalp{ij}{\gamma_3}|^2\bigr) +
\bigl(q_1|\skalp{ij}{\gamma_1}|^2+q_2|\skalp{ij}{\gamma_2}|^2\bigr) =
(\sqrt5+2)/(16\sqrt5) ( 3,1,1,3 )_{ij} + (\sqrt5-2)/(16\sqrt5) ( 1,3,3,1 )_{ij}$,
which is just $|\skalp{ij}{\psi}|^2=(\sqrt5+1,\sqrt5-1,\sqrt5-1,\sqrt5+1)_{ij}/(4\sqrt5)$.

The \textit{third example}
is a one-parameter family of weakly trivial nonproduct atomic CCS of the events $A$ and $B$ \eqref{eq:2x2AB},
which satisfies the law of total probability.
Recall that the product CCS of the events $A$ and $B$ \eqref{eq:2x2AB} given by~\eqref{eq:CCSclassU}, which was also used in the previous example,
could be transformed into the nonproduct CCS given by~\eqref{eq:CCStwist} by the diagonal unitary $V$, which makes it nonproduct.
One may see that the law of total probability~\eqref{eq:2x2LTP} is invariant to such a transformation, if the state vector is also transformed.
So transforming also the state vector~\eqref{eq:CCSclassULTPstate} used in the previous example by the same diagonal unitary $V$,
we get states which, together with the nonproduct atomic CCS~\eqref{eq:CCStwist} satisfies the law of total probability.
These state vectors are of the form
\begin{equation}
\label{eq:CCStwistLTPstate}
\cket{\psi} := \frac1{\sqrt2} \bigl(a \cket{00} + b \cket{01} + b \cket{10} + a \cket{11}\bigr),
\end{equation}
with the parameters $a,b\in\field{R}$, $a^2+b^2=1$.
These are entangled if and only if $b\neq\pm a$, which can be checked by~\eqref{eq:2x2SEP},
as $\det_{ij}(\skalp{ij}{\psi})=(a^2-b^2)/2$.
In these states, the events $A$ and $B$ are correlated if and only if $\abs{b}\neq\abs{a}$, which can be checked by~\eqref{eq:2x2COV},
as $\Delta_\psi(A,B)=\det_{ij}(|\skalp{ij}{\psi}|^2)=(\abs{a}^4-\abs{b}^4)/4 = (\abs{a}^2-\abs{b}^2)/4$.
The parameters $a,b$ for which the state~\eqref{eq:CCStwistLTPstate}
and the CCS~\eqref{eq:CCStwist} obey the law of total probability for the events $A$ and $B$
are exactly the same functions of the $c,s$ parameters of the CCS~\eqref{eq:CCStwist}
as in the previous example (see Appendix~\ref{appsec:2x2.CCSLTP}),
because of the invariance of the law of total probability~\eqref{eq:2x2LTP} to the transformation by the diagonal unitary $V$.
Also, the same transformation of the particular solution for the CCS~\eqref{eq:CCSclassUspec} with the state vector~\eqref{eq:CCSclassULTPstatespec} of the previous example
gives a particular solution here.
So let us have the vectors as the transformation of the vectors~\eqref{eq:CCSclassUspec} by the same diagonal unitary $V$
\begin{subequations}
\label{eq:CCStwistLTPspec}
\begin{align}
\cket{\gamma_0} &= \frac1{2\sqrt2} \bigl(   (\sqrt{2}+1) \cket{00} + \cket{01} + \cket{10} - (\sqrt{2}-1) \cket{11} \bigr),\\
\cket{\gamma_1} &= \frac1{2\sqrt2} \bigl( - \cket{00} + (\sqrt{2}+1) \cket{01} - (\sqrt{2}-1) \cket{10} - \cket{11} \bigr),\\
\cket{\gamma_2} &= \frac1{2\sqrt2} \bigl( - \cket{00} - (\sqrt{2}-1) \cket{01} + (\sqrt{2}+1) \cket{10} - \cket{11} \bigr),\\
\cket{\gamma_3} &= \frac1{2\sqrt2} \bigl(   (\sqrt{2}-1) \cket{00} - \cket{01} - \cket{10} - (\sqrt{2}+1) \cket{11} \bigr),
\end{align}
\end{subequations}
then $\set{C_k=\proj{\gamma_k}}_k$ is a noncommuting nonproduct atomic CCS of the events $A$ and $B$ \eqref{eq:2x2AB},
as follows from the general properties of the CCS given by~\eqref{eq:CCStwist},
although it may also be illustrative to check this directly by~\eqref{eq:2x2commutCCS}, \eqref{eq:2x2NP} and~\eqref{eq:2x2SCR}.
Let us also have the state vector as the transformation of the state vector given in~\eqref{eq:CCSclassULTPstatespec} by the same diagonal unitary $V$
\begin{equation}
\label{eq:CCStwitLTPspecpsi}
\cket{\psi} = \frac1{2\sqrt[4]{5}}
 \Bigl( \sqrt{\sqrt{5}+1}\cket{00} + \sqrt{\sqrt{5}-1}\cket{01} + \sqrt{\sqrt{5}-1}\cket{10} + \sqrt{\sqrt{5}+1}\cket{11}\Bigr).
\end{equation}
It is entangled, which can be checked by~\eqref{eq:2x2SEP},
as $\det_{ij}(\skalp{ij}{\psi})=1/(2\sqrt5)$.
In this state, the events $A$ and $B$ are correlated, which can be checked by~\eqref{eq:2x2COV},
as $\Delta_\psi(A,B)=\det_{ij}(|\skalp{ij}{\psi}|^2)=1/(4\sqrt5)$.
(This is also invariant to the transformation by $V$.)
The law of total probability~\eqref{eq:2x2LTP} holds for $\set{C_k=\proj{\gamma_k}}_k$ given by~\eqref{eq:CCStwistLTPspec}
in the state $\rho=\proj{\psi}$ given by~\eqref{eq:CCStwitLTPspecpsi},
which follows from the general construction (see Appendix~\ref{appsec:2x2.CCSLTP}),
although it may also be illustrative to check this directly by~\eqref{eq:2x2LTP},
which is basically the same calculation as in the end of the previous paragraph.

The \textit{fourth example}
is the one-parameter family of nontrivial (then nonproduct) atomic CCSs given by~\eqref{eq:CCSntrat},
in the system described by the state given by the state vector~\eqref{eq:Bell0}.
The law of total probability~\eqref{eq:2x2LTP} holds for these, which follows from the weak commutativity by Proposition~\ref{prop:CommLTP},
although it may also be illustrative to check this directly, by~\eqref{eq:2x2LTP},
as $q_1=q_2=0$, since $\skalp{\psi}{\gamma_1}=\skalp{\psi}{\gamma_2}=0$, see~\eqref{eq:CCSntrat}.

The \textit{fifth example}
is the one-parameter family of nontrivial nonatomic CCSs given by~\eqref{eq:CCS22ntrat},
in the system described by the state given by the state vector~\eqref{eq:Bell0}.
The law of total probability~\eqref{eq:2x2LTPpure} holds for these,
as $\cket{\psi_+}=C_+\cket{\psi}=\cket{\gamma_0}/\sqrt2=\cket{00}/\sqrt2$ and
$\cket{\psi_-}=C_-\cket{\psi}=\cket{\gamma_3}/\sqrt2=\cket{11}/\sqrt2$,
since
$\skalp{\psi}{\gamma_1}=\skalp{\psi}{\gamma_2}=0$, see~\eqref{eq:CCSntrat}.

% ******************************************************************************
% ******************************************************************************
\section{Indeterministic noncommuting common causes of perfect correlations}
\label{sec:CCSperfcorr}

In this section, we turn to the third question, prompted by the Maudlin-Werner debate: Can perfect correlations have indeterministic noncommuting common causes?
To recall, as Maudlin correctly observes, the EPR shows that any common cause accounting for a perfect correlation must be deterministic.
However, EPR's notion of the common cause is commuting. As we will show in this section, if we adopt also noncommuting common causes, then determinism will \textit{not} follow from perfect correlation. Thus, the EPR argument provides an argument for determinism only if we adopt classicality from the outset, that is, if we assume that both the correlating events and the common cause are in the very same commutative subalgebra. Otherwise, determinism does not follow. 
In the following, first we define \emph{maximal} and \emph{perfect correlations} in the general theory (see Section~\ref{sec:prob.gen}), in discrete finite systems (see Section~\ref{sec:prob.dxd}) and in two-qubit systems (see Section~\ref{sec:prob.2x2}) in turn.
Next, we provide some examples for indeterministic noncommuting CCSs of perfect correlations. 
All the examples are formulated in two-qubit systems (see Section~\ref{sec:prob.2x2}).

% ******************************************************************************
\subsection{Maximal and perfect correlations}
\label{sec:CCSperfcorr.theory}

\textit{In the general theory} (see Section~\ref{sec:prob.gen}),
we have already seen
that the correlation~\eqref{eq:Delta} of two events are bounded by $\pm1/4$, see~\eqref{eq:Deltabound}.
\begin{subequations}
The events $A$ and $B$ are \emph{maximally correlated} if $\Delta_\phi(A,B)=1/4$, which holds if and only if
\begin{equation}
\label{eq:maxcorr}
\phi(AB) = \phi(A^\perp B^\perp) = 1/2,
\end{equation}
then $\phi(AB^\perp)=\phi(A^\perp B)=0$ also holds,
so $\phi(AB)=\phi(A^\perp B^\perp)=\phi(A)=\phi(B)=\phi(A^\perp)=\phi(B^\perp)=1/2$.
The events $A$ and $B$ are \emph{maximally anticorrelated} if $\Delta_\phi(A,B)=-1/4$, which holds if and only if
\begin{equation}
\label{eq:maxanticorr}
\phi(AB^\perp)=\phi(A^\perp B)=1/2,
\end{equation}
then $\phi(AB) = \phi(A^\perp B^\perp) = 0$ also holds,
so $\phi(AB^\perp)=\phi(A^\perp B)=\phi(A)=\phi(B)=\phi(A^\perp)=\phi(B^\perp)=1/2$.
\end{subequations}
The events $A$ and $B$ are \emph{perfectly correlated} if
\begin{subequations}
\begin{equation}
\label{eq:perfcorr}
\phi(AB) = \phi(A) = \phi(B),
\end{equation}
(equivalently $\phi(A^\perp B^\perp) = \phi(A^\perp) = \phi(B^\perp)$)
then \emph{equivalently} $\phi(AB^\perp)=\phi(A^\perp B)=0$ also holds,
but $\phi(A^\perp B^\perp)\neq\phi(AB)$ in general,
so maximally correlated events are perfectly correlated, but the reverse is not true.
The correlation of perfectly correlated events is 
$\Delta_\phi(A,B)=\phi(AB)(1-\phi(AB))=\phi(AB)\phi(A^\perp B^\perp)\in[0,1/4]$,
or the same formula with $\phi(A)$ or $\phi(B)$.
We also have the equivalent condition $\phi(AB+A^\perp B^\perp)=1$ for perfect correlation.
The events $A$ and $B$ are \emph{perfectly anticorrelated} if
\begin{equation}
\label{eq:perfanticorr}
\phi(AB^\perp) = \phi(A) = \phi(B^\perp),
\end{equation}
\end{subequations}
(equivalently $\phi(A^\perp B) = \phi(A^\perp) = \phi(B)$)
then \emph{equivalently} $\phi(AB)=\phi(A^\perp B^\perp)=0$ also holds,
but $\phi(AB^\perp)\neq\phi(A^\perp B)$ in general,
so maximally anticorrelated events are perfectly anticorrelated, but the reverse is not true.
The correlation of perfectly anticorrelated events is 
$\Delta_\phi(A,B)=-\phi(AB^\perp)(1-\phi(AB^\perp))=-\phi(AB^\perp)\phi(A^\perp B)\in[-1/4,0]$,
or the same formula with $\phi(A)$ or $\phi(B^\perp)$.
We also have the equivalent condition $\phi(AB^\perp+A^\perp B)=1$ for perfect anticorrelation.

It is clear that the maximal or perfect anticorrelation of events $A$ and $B$
is just the maximal or perfect correlation of events $A$ and $B^\perp$,
so, from now on, we deal only with maximal or perfect correlation.

\begin{prop}
\label{prop:classmaxodet}
In the general theory (see Section~\ref{sec:prob.gen}),
if the events $A,B\in\mathcal{P}(\mathcal{N})$ are perfectly correlated~\eqref{eq:perfcorr}
and they possess a \emph{weakly commuting} CCS~\eqref{eq:commutCCS},
then the CCS is deterministic~\eqref{eq:determCCS}.
\end{prop}

\begin{proof}
Let us have the two projections
$C:=\sum_{k:\phi(C_k)\neq0} C_k$ and $D:=\sum_{k:\phi(C_k)=0} C_k$.
For perfect correlation~\eqref{eq:perfcorr} we have $\phi(AB^\perp)=0$.
Then $0=\phi(AB^\perp)=\phi(AB^\perp C)+\phi(AB^\perp D)$,
where the second term is zero by Lemma~\ref{lem:wcomm}.
Then we have
$0= \phi(AB^\perp C) =\phi\bigl(AB^\perp\sum_{k:\phi(C_k)\neq0}C_k\bigr)
=\sum_{k:\phi(C_k)\neq0}\phi( AB^\perp C_k^2)=\sum_{k:\phi(C_k)\neq0} \phi\bigl(C_kAB^\perp C_k\bigr)$,
where at the last step the weak commutativity~\eqref{eq:wcommutCCS} of $\set{C_k}_k$ was used.
Then we have that $\sum_{k:\phi(C_k)\neq0} \phi(C_kAB^\perp C_k)=0$,
which is a sum of the nonnegative terms $\phi(C_kAB^\perp C_k)\geq0$,
so $\phi(C_kAB^\perp C_k)=0$ if $\phi(C_k)\neq0$.
Then
$\phi(C_kAC_k)=\phi(C_kABC_k)+\phi(C_kAB^\perp C_k)=\phi(C_kABC_k)$ if $\phi(C_k)\neq0$,
leading to
$\phi(A|C_k)=\phi(AB|C_k)$ if $\phi(C_k)\neq0$ by~\eqref{eq:condphi}.
Switching the roles of $A$ and $B$ we also have
$\phi(B|C_k) = \phi(AB|C_k)$ if $\phi(C_k)\neq0$.
Then the~\eqref{eq:screen.original} expression of the screening-off leads to the deterministicity~\eqref{eq:determCCS.AB}.
\end{proof}

Note that this proof relies on the weak commutativity of the CCS~\eqref{eq:commutCCS},
since for noncommuting CCSs $0=\sum_{k:\phi(C_k)\neq0}\phi( AB^\perp C_k^2)\neq
\sum_{k:\phi(C_k)\neq0} \phi\bigl(C_kAB^\perp C_k\bigr)$ in general.

\textit{In discrete finite systems} (see Section~\ref{sec:prob.dxd}),
the projections $A,B$ are perfectly correlated~\eqref{eq:perfcorr} in a state given by $\rho$,
if and only if
\begin{subequations}
\label{eq:dxdperfcorrABstate}
\begin{equation}
\label{eq:dxdperfcorrABstate.rho}
(AB+A^\perp B^\perp)\rho(AB+A^\perp B^\perp)=\rho,
\end{equation}
or, in a pure state given by the state vector $\cket{\psi}$,
if and only if
\begin{equation}
\label{eq:dxdperfcorrABstate.psi}
(AB+A^\perp B^\perp)\cket{\psi}=\cket{\psi}.
\end{equation}
\end{subequations}

\textit{In two-qubit systems} (see Section~\ref{sec:prob.2x2}), 
the states giving perfect correlation for $A$ and $B$~\eqref{eq:2x2AB}
are given by the density operators~\eqref{eq:dxdperfcorrABstate.rho} of the form
\begin{subequations}
\label{eq:2x2perfcorrABstate}
\begin{equation}
\label{eq:2x2perfcorrABstate.gen}
    \rho = \frac12\bigl((1+r_3)\proj{00} + (r_1-ir_2) \cket{00}\bra{11}+ 
           (r_1+ir_2) \cket{11}\bra{00} + (1-r_3)\proj{11}\bigr)
\end{equation}
with the parameters $r_1,r_2,r_3\in\field{R}$, $r_1^2+r_2^2+r_3^2 \leq 1$.
Then $\phi(AB)=\phi(A)=\phi(B)=\frac12(1+r_3)$ and
$\Delta(A,B)=(1-r_3^2)/4$, being maximal for $r_3=0$.
The state described by $\rho$ is pure if and only if $r_1^2+r_2^2+r_3^2 = 1$,
then it is given by the state vector~\eqref{eq:dxdperfcorrABstate.psi} of the form
\begin{equation}
\label{eq:2x2perfcorrABstate.pure}
    \cket{\psi} = x\cket{00} + y\cket{11}
\end{equation}
with the parameters $x,y\in\field{C}$, $\abs{x}^2+\abs{y}^2=1$.
\end{subequations}
Then $\phi(AB)=\phi(A)=\phi(B)=\abs{x}^2$ and
$\Delta(A,B)=\abs{x}^2(1-\abs{x}^2)=\abs{y}^2(1-\abs{y}^2)=\abs{x}^2\abs{y}^2$, being maximal for $\abs{x}=\abs{y}=1/\sqrt{2}$.

% ******************************************************************************
\subsection{Examples}
\label{sec:CCSperfcorr.examples}

Now, we provide some examples for indeterministic CCSs of perfect correlations in two-qubit systems (see Section~\ref{sec:prob.2x2}). 
We do not construct new examples here, we just go through the examples already constructed in Section~\ref{sec:CCSprod},
from the point of view of the deterministicity of the CCS.
(See also Table~\ref{tab:examples} in Appendix~\ref{appsec:2x2.summary} for convenience.)
Recall that the conditional probabilities~\eqref{eq:2x2phiCkABpure}-\eqref{eq:2x2phiCkABatomic} make sense only for conditions of nonzero probability.
Thus, even for trivial atomic CCSs, where the screening-off is independent of the state,
the conditional probabilities still depend on the state, although only through the condition $\bra{\gamma_k}\rho\cket{\gamma_k}\neq0$,
see~\eqref{eq:2x2phiCkABatomic} in Lemma~\ref{lem:2x2condphiSpec}.
Our last four examples are the state specific ones, which are about maximal \eqref{eq:maxcorr} (hence perfect \eqref{eq:perfcorr}) correlation.

The \textit{first example} 
is the one-parameter family of product atomic (then strongly trivial) CCSs given by~\eqref{eq:CCSclassU}.
Here, by~\eqref{eq:2x2phiCkABatomic}, we have
\begin{equation}
\label{eq:CCSclassU.phiCkAB}
\phi(A|C_k)=(c^2,c^2,s^2,s^2)_k,\qquad 
\phi(B|C_k)=(c^2,s^2,c^2,s^2)_k,
\end{equation}
so this CCS is deterministic if and only if $\theta\in \pi\field{Z}$, otherwise it is indeterministic.
These are exactly the parameter values for which it is commuting.
The particular case~\eqref{eq:CCSclass} is $\theta\in\pi\field{Z}$, then $\phi(A|C_k)=(1,1,0,0)_k$ and $\phi(B|C_k)=(1,0,1,0)_k$,
so the CCS~\eqref{eq:CCSclass} is deterministic.
The other particular case~\eqref{eq:CCSGabor} is $\theta\in\pi(\field{Z}+1/2)$, then $\phi(A|C_k)=\phi(B|C_k)=(1/2,1/2,1/2,1/2)_k$
so the CCS~\eqref{eq:CCSGabor} is maximally indeterministic.
Recall again that, by~\eqref{eq:2x2phiCkABatomic} in Lemma~\ref{lem:2x2condphiSpec},
we have these independently of the state, as far as $\bra{\gamma_k}\rho\cket{\gamma_k}\neq0$,
which fails among the states~\eqref{eq:2x2perfcorrABstate} giving perfect correlation for $A$ and $B$ only in a subset of zero measure.
(For those states, for the given $k$, the conditional probabilities $\phi(A|C_k)$ and $\phi(B|C_k)$ are undefined, and the values written above are invalid.)

The \textit{second example}
is the one-parameter family of weakly trivial nonproduct atomic CCSs given by~\eqref{eq:CCStwist}.
Here, by~\eqref{eq:2x2phiCkABatomic}, we have
\begin{equation}
\label{eq:CCStwist.phiCkAB}
\phi(A|C_k)=(c^2,c^2,s^2,s^2)_k,\qquad 
\phi(B|C_k)=(c^2,s^2,c^2,s^2)_k,
\end{equation}
so this CCS is deterministic if and only if $\theta\in \pi\field{Z}$, otherwise it is indeterministic.
These conditional probabilities are the same as~\eqref{eq:CCSclassU.phiCkAB} for the CCS given by~\eqref{eq:CCSclassU},
since these are of the form~\eqref{eq:condphiatomic},
which is invariant to the transformation by the diagonal unitary $V$.
(Indeed,
$\phi(X|(V\proj{\gamma_k}V^\dagger))=\bra{\gamma_k}V^\dagger X V\cket{\gamma_k}=\bra{\gamma_k}X\cket{\gamma_k}$
if $[X,V]=0$, which holds for the subalgebra generated by $A$ and $B$.)
However, here we have now a nonproduct CCS if $\theta\not\in\pi\field{Z}$.
Recall again that, by~\eqref{eq:2x2phiCkABatomic} in Lemma~\ref{lem:2x2condphiSpec},
we have these independently of the state, as far as $\bra{\gamma_k}\rho\cket{\gamma_k}\neq0$,
which fails among the states~\eqref{eq:2x2perfcorrABstate} giving perfect correlation for $A$ and $B$ only in a subset of zero measure.
(For those states, for the given $k$, the conditional probabilities $\phi(A|C_k)$ and $\phi(B|C_k)$ are undefined, and the values written above are invalid.)

The \textit{third example}
is the one-parameter family of weakly trivial nonproduct atomic CCSs given by~\eqref{eq:CCSBell}.
Here, by~\eqref{eq:2x2phiCkABatomic}, we have
\begin{equation}
\label{eq:CCSBell.phiCkAB}
\phi(A|C_k)=(c^2,c^2,s^2,s^2)_k,\qquad
\phi(B|C_k)=(1/2,1/2,1/2,1/2)_k,
\end{equation}
so this CCS is never deterministic, although the conditional probabilities of $A$ can be sharp for $\theta\in \pi\field{Z}$.
If $\theta\in\pi(\field{Z}+1/2)$, then $\phi(A|C_k)=\phi(B|C_k)=(1/2,1/2,1/2,1/2)_k$.
Recall again that, by~\eqref{eq:2x2phiCkABatomic} in Lemma~\ref{lem:2x2condphiSpec},
we have these independently of the state, as far as $\bra{\gamma_k}\rho\cket{\gamma_k}\neq0$,
which fails among the states~\eqref{eq:2x2perfcorrABstate} giving perfect correlation for $A$ and $B$ only in a subset of zero measure.
(For those states, for the given $k$, the conditional probabilities $\phi(A|C_k)$ and $\phi(B|C_k)$ are undefined, and the values written above are invalid.)

The \textit{fourth example}
is the one-parameter family of weakly trivial nonproduct atomic CCSs given by~\eqref{eq:CCShyper}.
Here, by~\eqref{eq:2x2phiCkABatomic}, we have
\begin{equation}
\label{eq:CCShyper.phiCkAB}
\begin{split}
\phi(A|C_k)&=N^2(\ee^{+\xi}+\ee^{+\zeta},\ee^{-\xi}+\ee^{-\zeta},\ee^{+\xi}+\ee^{+\zeta},\ee^{-\xi}+\ee^{-\zeta})_k,\\
\phi(B|C_k)&=N^2(\ee^{+\xi}+\ee^{-\zeta},\ee^{-\xi}+\ee^{+\zeta},\ee^{-\xi}+\ee^{+\zeta},\ee^{+\xi}+\ee^{-\zeta})_k,
\end{split}
\end{equation}
so this CCS is never deterministic, as $N^2=1/(\ee^{+\xi}+\ee^{+\zeta}+\ee^{-\xi}+\ee^{-\zeta})$.
If $\xi=\zeta=0$, then $\phi(A|C_k)=\phi(B|C_k)=(1/2,1/2,1/2,1/2)_k$.
Recall again that, by~\eqref{eq:2x2phiCkABatomic} in Lemma~\ref{lem:2x2condphiSpec},
we have these independently of the state, as far as $\bra{\gamma_k}\rho\cket{\gamma_k}\neq0$,
which fails among the states~\eqref{eq:2x2perfcorrABstate} giving perfect correlation for $A$ and $B$ only in a subset of zero measure.
(For those states, for the given $k$, the conditional probabilities $\phi(A|C_k)$ and $\phi(B|C_k)$ are undefined, and the values written above are invalid.)

The \textit{fifth example}
is the one-parameter family of nontrivial (then nonproduct) atomic CCSs given by~\eqref{eq:CCSntrat}
in the system described by the state given by the state vector~\eqref{eq:Bell0},
in which $A$ and $B$ are maximally correlated~\eqref{eq:maxcorr}, so also perfectly correlated~\eqref{eq:perfcorr}.
Here, by~\eqref{eq:2x2phiCkABatomic}, we have
\begin{equation}
\label{eq:CCSntrat.phiCkAB}
\phi(A|C_k)=(1,-,-,0)_k,\qquad
\phi(B|C_k)=(1,-,-,0)_k,
\end{equation}
so this CCS is deterministic.
Note that the values for $k\in\set{1,2}$ are undefined, since $\phi(C_k)=0$ then.

The \textit{sixth example}
is the one-parameter family of nontrivial (then nonproduct) atomic CCSs given by~\eqref{eq:CCSntratU}
in the system described by the state given by the state vector~\eqref{eq:Bell0}
in which $A$ and $B$ are maximally correlated~\eqref{eq:maxcorr}, so also perfectly correlated~\eqref{eq:perfcorr}.
Here, by~\eqref{eq:2x2phiCkABatomic}, we have
\begin{equation}
\label{eq:CCSntratU.phiCkAB}
\phi(A|C_k)=(1/2,-,-,1/2)_k,\qquad
\phi(B|C_k)=(1/2,-,-,1/2)_k,
\end{equation}
so this CCS is indeterministic.
Note that the values for $k\in\set{1,2}$ are undefined, since $\phi(C_k)=0$ then.

The \textit{seventh example} 
is the one-parameter family of nontrivial nonatomic (rank-two) CCSs given by~\eqref{eq:CCS22ntrat}
in the system described by the state given by the state vector~\eqref{eq:Bell0}
in which $A$ and $B$ are maximally correlated~\eqref{eq:maxcorr}, so also perfectly correlated~\eqref{eq:perfcorr}.
Here, by~\eqref{eq:2x2phiCkABpure}, we have
\begin{equation}
\label{eq:CCS22ntrat.phiCkAB}
\phi(A|C_k)=(1,0)_k,\qquad
\phi(B|C_k)=(1,0)_k,
\end{equation}
as $\cket{\psi_+}=C_+\cket{\psi}=\cket{\gamma_0}/\sqrt2=\cket{00}/\sqrt2$ and
$\cket{\psi_-}=C_-\cket{\psi}=\cket{\gamma_3}/\sqrt2=\cket{11}/\sqrt2$,
since
$\skalp{\gamma_1}{\psi}=\skalp{\gamma_2}{\psi}=0$, see~\eqref{eq:CCSntrat},
so this CCS is deterministic.

The \textit{eighth example}
is the one-parameter family of nontrivial nonatomic (rank-two) CCSs given by~\eqref{eq:CCS22ntratU}
in the system described by the state given by the state vector~\eqref{eq:Bell0}
in which $A$ and $B$ are maximally correlated~\eqref{eq:maxcorr}, so also perfectly correlated~\eqref{eq:perfcorr}.
Here, by~\eqref{eq:2x2phiCkABpure}, we have
\begin{equation}
\label{eq:CCS22ntratU.phiCkAB}
\phi(A|C_k)=(1/2,1/2)_k,\qquad 
\phi(B|C_k)=(1/2,1/2)_k,
\end{equation}
as $\cket{\psi_+}=C_+\cket{\psi}=\cket{\gamma_0}/\sqrt2=\cket{{+}{+}}/\sqrt2$ and
$\cket{\psi_-}=C_-\cket{\psi}=\cket{\gamma_3}/\sqrt2=\cket{{-}{-}}/\sqrt2$,
since
$\skalp{\gamma_1}{\psi}=\skalp{\gamma_2}{\psi}=0$, see~\eqref{eq:CCSntrat},
so this CCS is maximally indeterministic.

% ******************************************************************************
% ******************************************************************************
\section{Limitations of atomic common cause systems for perfect correlations}
\label{sec:limitations}

In the previous three sections, we separately analyzed three features of the noncommuting CCSs.
In Section~\ref{sec:CCSprod}, we have shown examples for nontrivial noncommuting CCSs;
in Section~\ref{sec:CCSLTP} for noncommuting CCSs satisfying the law of total probability;
and in Section~\ref{sec:CCSperfcorr} for noncommuting indeterministic CCSs accounting for perfect correlations.
In this section, we consider the interplay between the above properties
and we show that the law of total probability in case of perfect correlation leads to the deterministicity of the CCS,
and also to the (weak) commutativity of the \emph{atomic} CCS.
All the discussions and Propositions will be formulated for general systems (see Section~\ref{sec:prob.gen}) and discrete finite systems (see Section~\ref{sec:prob.dxd}).
Next, we provide some examples for the illustration of the deterministicity and the commutativity of CCSs in these cases.
All the examples are formulated in two-qubit systems (see Section~\ref{sec:prob.2x2}).

% ******************************************************************************
\subsection{Deterministicity and commutativity for perfect correlations}
\label{sec:limitations.theory}

In the next two Propositions, we show how the prescription of the law of total probability confines the CCSs to be deterministic and, if atomic, to be weakly commutative~\eqref{eq:wcommutCCS},
that is, it can contain noncommuting causes $C_k$ of zero probability only.

\begin{prop}
\label{prop:CCSLTPPCdeterm}
In the general theory (see Section~\ref{sec:prob.gen}), 
for states giving perfect correlation for $A$ and $B$ \eqref{eq:perfcorr},
the CCSs $\set{C_k}_k$
obeying the law of total probability~\eqref{eq:LTP}
are deterministic~\eqref{eq:determCCS}.
\end{prop}

Note that Proposition~\ref{prop:classmaxodet} follows from this by Proposition~\ref{prop:CommLTP}.

\begin{proof}
Let us use the notation $A_0:=A$, $A_1:=A^\perp$, $B_0:=B$, $B_1:=B^\perp$,
then the law of total probability~\eqref{eq:LTP.gen} takes the form 
\begin{equation}
\label{eq:LTPij}
\sum_{k\in K, \phi(C_k)\neq0} \phi(C_k)\phi(A_iB_j|C_k) = \phi(A_iB_j)\quad\text{for $i,j\in\set{0,1}$}.
\end{equation}
Recall that $\phi(AB^\perp)=0$ and $\phi(A^\perp B)=0$ for states describing perfect correlations~\eqref{eq:perfcorr},
so the right-hand side of~\eqref{eq:LTPij} is $(\phi(AB),0,0,\phi(A^\perp B^\perp))_{ij}$.
Since the left-hand side of~\eqref{eq:LTPij} is a convex combination (with weights $\phi(C_k)\geq0$, $\sum_k\phi(C_k)=1$)
of nonnegative terms $\phi(A_iB_j|C_k)\geq0$,
we have for these that 
\begin{equation}
\phi(A_iB_j|C_k) = \bigl(\phi(AB|C_k),0,0,\phi(A^\perp B^\perp|C_k)\bigr)_{ij}
\end{equation}
for all $k\in K$ for which $\phi(C_k)\neq0$.
Then the right-hand side of the screening-off condition~\eqref{eq:screen.balanced} is zero, which enforces that
also $\phi(A^\perp B^\perp|C_k)=0$ or $\phi(AB|C_k)=0$ on the left-hand side for all $k\in K$ for which $\phi(C_k)\neq0$,
so, using $\Id=\sum_{ij}A_iB_j$, we have either
\begin{equation}
\label{eq:CCSLTPPCcases}
\phi(A_iB_j|C_k)=(1,0,0,0)_{ij} \quad\text{or}\quad
\phi(A_iB_j|C_k)=(0,0,0,1)_{ij}
\end{equation}
for all $k\in K$ for which $\phi(C_k)\neq0$,
which is just the deterministicity~\eqref{eq:determCCS.X}.
\end{proof}

\begin{prop}
\label{prop:aCCSLTPPCComm}
For discrete finite systems (see Section~\ref{sec:prob.dxd}), 
for states giving perfect correlation for $A$ and $B$ \eqref{eq:perfcorr},
the \emph{atomic} CCSs $\set{C_k}_k$
obeying the law of total probability~\eqref{eq:LTP}
are weakly commuting~\eqref{eq:wcommutCCS}.
\end{prop}

\begin{proof}
Let us use the notation $A_0:=A$, $A_1:=A^\perp$, $B_0:=B$, $B_1:=B^\perp$.
Recall from the proof of Proposition~\ref{prop:CCSLTPPCdeterm} that
for states giving perfect correlation for $A$ and $B$ \eqref{eq:perfcorr},
for the CCSs $\set{C_k}_k$
obeying the law of total probability~\eqref{eq:LTP} we have the two cases in~\eqref{eq:CCSLTPPCcases}.
It holds in general that $\Ran(C_k A_iB_j C_k)\subseteq\Ran(C_k)$.
Since the CCS is now also \emph{atomic},
we also have that
\begin{subequations}
\label{eq:atomicrelations}
\begin{equation}
\label{eq:atomicrelations.proj}
C_k A_iB_j C_k = c_{ijk} C_k
\end{equation}
for a $c_{ijk}\in[0,1]$, so 
\begin{equation}
\label{eq:atomicrelations.prob}
\phi(C_k A_iB_j C_k)=c_{ijk}\phi(C_k).
\end{equation}
\end{subequations}
If $\phi(C_k)\neq0$ then, by the definition of the conditional state~\eqref{eq:condphi},
the two cases in~\eqref{eq:CCSLTPPCcases} are
\begin{equation}
c_{ijk}=(1,0,0,0)_{ij} \quad\text{or}\quad
c_{ijk}=(0,0,0,1)_{ij}
\end{equation}
for all $k\in K$ by~\eqref{eq:atomicrelations.prob}.
So if $\phi(C_k)\neq0$ then either
\begin{equation}
C_k A_iB_j C_k=(C_k,0,0,0)_{ij} \quad\text{or}\quad
C_k A_iB_j C_k=(0,0,0,C_k)_{ij}
\end{equation}
for all $k\in K$ by~\eqref{eq:atomicrelations.proj}.
This leads to that if $\phi(C_k)\neq0$ then
either $\Ran(C_k)\subseteq\Ran(AB)$ or $\Ran(C_k)\subseteq\Ran(A^\perp B^\perp)$,
then $[C_k,A_iB_j]=0$ for all $i,j\in\set{0,1}$,
from which the weak commutativity follows by~\eqref{eq:wcommutCCS.X}.
\end{proof}

Note that the~\eqref{eq:atomicrelations} consequences of the atomicity of the CCS
are the key steps allowing us to achieve the result about commutativity
(which is the property of the ``whole'' \emph{projections} $A$, $B$ and $\set{C_k}_k$),
from the assumptions of screening-off~\eqref{eq:screen}, the law of total probability~\eqref{eq:LTP} and the perfect correlation~\eqref{eq:perfcorr}
(which are the properties of the \emph{expectation values} of the projections only).

For the convenience of the reader,
we also formulate the above Proposition and its proof in Dirac's bra-ket formalism.

\begin{prop}
\label{prop:aCCSLTPPCComm2}
For discrete finite systems (see Section~\ref{sec:prob.dxd}), 
for states giving perfect correlation for $A$ and $B$ \eqref{eq:dxdperfcorrABstate.rho},
the \emph{atomic} CCSs $\set{C_k=\proj{\gamma_k}}_k$
obeying the law of total probability~\eqref{eq:LTPrhovec}
are weakly commuting~\eqref{eq:wcommutCCS}.
\end{prop}

\begin{proof}
Let us use the notation $A_0:=A$, $A_1:=A^\perp$, $B_0:=B$, $B_1:=B^\perp$,
then the law of total probability~\eqref{eq:LTPrhovec} takes the form
\begin{equation}
\label{eq:LTPrhovecij}
\sum_{k\in K} q_k \bra{\gamma_k} A_iB_j \cket{\gamma_k} = \Tr(\rho A_iB_j)\quad\text{for $i,j\in\set{0,1}$}.
\end{equation}
For states describing perfect correlations~\eqref{eq:dxdperfcorrABstate.rho},
the right-hand side of~\eqref{eq:LTPrhovecij} above is $(\Tr(\rho AB),0,0,\Tr(\rho A^\perp B^\perp))_{ij}$.
Since the left-hand side of~\eqref{eq:LTPrhovecij} is
a convex combination (with weights $q_k=\bra{\gamma_k}\rho\cket{\gamma_k}\geq0$, $\sum_kq_k=1$)
of nonnegative terms $\bra{\gamma_k}A_iB_j\cket{\gamma_k}\geq0$,
we have for these that 
\begin{equation}
\bra{\gamma_k}A_iB_j\cket{\gamma_k} = \bigl(\bra{\gamma_k}AB\cket{\gamma_k},0,0,\bra{\gamma_k}A^\perp B^\perp\cket{\gamma_k}\bigr)_{ij} 
\end{equation}
for all $k\in K$ for which $q_k\neq0$.
Then the right-hand side of the screening-off condition~\eqref{eq:dxdscreenatomic.balanced} is zero, which enforces that also
$\bra{\gamma_k}A^\perp B^\perp\cket{\gamma_k}=0$ or $\bra{\gamma_k}AB\cket{\gamma_k}=0$ on the left-hand side for all $k\in K$ for which $q_k\neq0$,
so, using $\Id=\sum_{ij}A_iB_j$, we have either
\begin{equation}
\bra{\gamma_k}A_iB_j\cket{\gamma_k}=(1,0,0,0)_{ij} \quad\text{or}\quad
\bra{\gamma_k}A_iB_j\cket{\gamma_k}=(0,0,0,1)_{ij}
\end{equation}
for all $k\in K$ for which $q_k\neq0$.
This leads to that if $q_k\neq0$ then $\cket{\gamma_k}=AB\cket{\gamma_k}$ or $\cket{\gamma_k}=A^\perp B^\perp \cket{\gamma_k}$,
then $[C_k,A_iB_j]=0$ for all $i,j\in\set{0,1}$,
from which the weak commutativity follows by~\eqref{eq:wcommutCCS.X}.
\end{proof}

Proposition~\ref{prop:aCCSLTPPCComm} (Proposition~\ref{prop:aCCSLTPPCComm2}) somewhat weakens the concept of noncommuting common causes, since it shows that the law of total probability forces the CCSs to be ``almost'' commutative---``almost'' in the sense that all the noncommuting elements of the CCS should get zero probability. Recall, however, that Proposition~\ref{prop:aCCSLTPPCComm} holds only for atomic CCSs, and only for perfect correlations, thus its scope is rather limited. Hence, the natural question arises whether noncommuting nonatomic CCSs for perfect correlations, obeying the law of total probability are excluded. Below we show that this is not the case (see the \textit{second example} below).

% ******************************************************************************
\subsection{Examples}
\label{sec:limitations.examples}

Now we provide some examples in two-qubit systems (see Section~\ref{sec:prob.2x2})
to illustrate that the law of total probability in case of perfect correlation leads to 
the weak commutativity of the CCS in the atomic case,
but not in the nonatomic case.
(See also Table~\ref{tab:examples} in Appendix~\ref{appsec:2x2.summary} for convenience.)

The \textit{first example}
is a family of nontrivial atomic CCSs of the events $A$ and $B$ \eqref{eq:2x2AB}
which satisfies the law of total probabilities
in the system described by the states~\eqref{eq:2x2perfcorrABstate.gen} giving perfect correlations.
The immediate consequence of Proposition~\ref{prop:aCCSLTPPCComm}
for two-qubit systems (see Section~\ref{sec:prob.2x2}) is that,
for states $\rho$ giving perfect correlation~\eqref{eq:2x2perfcorrABstate.gen} of the events $A$ and $B$~\eqref{eq:2x2AB},
the atomic CCSs $\set{C_k=\proj{\gamma_k}}_k$
obeying the law of total probability~\eqref{eq:2x2LTP}
can differ from the classical CCS \eqref{eq:CCSclass}
only in causes of zero probability.
These CCSs are exactly those which are given by the vectors
\begin{subequations}
\label{eq:CCSntratC}
\begin{align}
    \cket{\gamma_0}&:=\cket{00},& \abs{\bra{\gamma_0}\rho\cket{\gamma_0}}^2&=(1+r_3)/2,\\
    \cket{\gamma_1}&:=c\cket{01}+s\cket{10},& \abs{\bra{\gamma_1}\rho\cket{\gamma_1}}^2&=0,\\
    \cket{\gamma_2}&:=-\Cc{s}\cket{01}+\Cc{c}\cket{10},& \abs{\bra{\gamma_2}\rho\cket{\gamma_2}}^2&=0,\\
    \cket{\gamma_3}&:=\cket{11},& \abs{\bra{\gamma_3}\rho\cket{\gamma_3}}^2&=(1-r_3)/2,
\end{align}
\end{subequations}
with $c,s\in\field{C}$, $\abs{c}^2+\abs{s}^2=1$,
covering also the classical case~\eqref{eq:CCSclass} for $c=0$ or $s=0$ up to phases and ordering,
and also the CCS given by~\eqref{eq:CCSntrat} for $c,s\in\field{R}$, where the angle parametrization was used.
It is straightforward to check that $\set{C_k=\proj{\gamma_k}}_k$ is an atomic partition.
It consists of nonproduct projections if and only if $c\neq0$ and $s\neq0$, which can be checked by~\eqref{eq:2x2NP},
as $\det_{ij}(\skalp{ij}{\gamma_1})=-\Cc{\det_{ij}(\skalp{ij}{\gamma_2})}=-cs$.
So it is a strongly trivial CCS if $c=0$ or $s=0$.
It is a CCS of the events $A$ and $B$~\eqref{eq:2x2AB} for the state~\eqref{eq:2x2perfcorrABstate.gen},
which can be checked by~\eqref{eq:2x2SCR},
as $\det_{ij}(|\skalp{ij}{\gamma_0}|^2)=\det_{ij}(|\skalp{ij}{\gamma_3}|^2)=0$,
and although $\det_{ij}(|\skalp{ij}{\gamma_1}|^2)=\det_{ij}(|\skalp{ij}{\gamma_2}|^2)=-\abs{cs}^2$,
$C_1$ and $C_2$ are of zero probability,
$\abs{\bra{\gamma_1}\rho\cket{\gamma_1}}^2=\abs{\bra{\gamma_2}\rho\cket{\gamma_2}}^2=0$
because of the perfect correlation~\eqref{eq:2x2perfcorrABstate.pure}.
This does not hold for any state, where $A$ and $B$ are not perfectly correlated,
so $\set{C_k=\proj{\gamma_k}}_k$ is a nontrivial CCS if $c\neq0$ and $s\neq0$.
It is a commuting CCS of the events $A$ and $B$ \eqref{eq:2x2AB} if and only if $c=0$ or $s=0$, by~\eqref{eq:2x2commutCCS},
otherwise it is weakly commuting~\eqref{eq:wcommutCCS},
since the commutativity~\eqref{eq:commutCCS} fails only for $C_1$ and $C_2$, which are of zero probability.
The law of total probability holds for this CCS, which follows from the weak commutativity by Proposition~\ref{prop:CommLTP},
although it may also be illustrative to check this directly, by~\eqref{eq:2x2LTPatomic},
since $\rho\cket{\gamma_1}=\rho\cket{\gamma_2}=0$.
Note that this is a deterministic CCS, by Proposition~\ref{prop:CCSLTPPCdeterm},
although it may also be illustrative to check this directly, by~\eqref{eq:2x2phiCkABatomic}.

The \textit{second example}
is a family of nontrivial nonatomic CCSs of the events $A$ and $B$ \eqref{eq:2x2AB}
which satisfies the law of total probabilities
in the system described by the states~\eqref{eq:2x2perfcorrABstate.gen} giving perfect correlations.
Note that atomisticity of the CCS was a key property in Proposition~\ref{prop:aCCSLTPPCComm}.
For nonatomic CCSs $\set{C_k}_k$ of the events $A$ and $B$, Proposition~\ref{prop:aCCSLTPPCComm} does not hold,
that is, perfect correlation~\eqref{eq:dxdperfcorrABstate.psi} for the events $A$ and $B$~\eqref{eq:2x2AB}
together with the law of total probability~\eqref{eq:LTPrho}
do not lead to the weak commutativity of the CCS.
The simplest example for this is given in two-qubit systems (see Section~\ref{sec:prob.2x2}) by
\begin{equation}
\label{eq:CCS22ntratC}
    C_+ := \proj{\gamma_0} + \proj{\gamma_1},\qquad
    C_- := \proj{\gamma_2} + \proj{\gamma_3},\qquad
    \Tr(\rho C_\pm)=(1\pm r_3)/2,
\end{equation}
with the $\cket{\gamma_k}$ vectors given in \eqref{eq:CCSntratC},
covering also the weakly trivial CCSs given in~\eqref{eq:CCStrivAB.2} for $c=0$ and for $s=0$,
and the CCS given by~\eqref{eq:CCS22ntrat} for $c,s\in\field{R}$, where the angle parametrization was used.
So $\set{C_+,C_-}$ is a weakly trivial CCS if $c=0$ or $s=0$.
It is a rank-two CCS of the events $A$ and $B$~\eqref{eq:2x2AB} for the state~\eqref{eq:2x2perfcorrABstate.gen},
which can be checked by~\eqref{eq:dxdscreen.balanced}
since
$\rho\cket{\gamma_1}=\rho\cket{\gamma_2}=0$,
because of the perfect correlation~\eqref{eq:2x2perfcorrABstate.pure}.
This does not hold for any state, where $A$ and $B$ are not perfectly correlated,
so $\set{C_k=\proj{\gamma_k}}_k$ is a nontrivial CCS if $c\neq0$ and $s\neq0$.
It is a commuting CCS of the events $A$ and $B$ \eqref{eq:2x2AB} if and only if $c=0$ or $s=0$, by~\eqref{eq:2x2commutCCS},
otherwise it is noncommuting.
The law of total probability holds for this CCS, by~\eqref{eq:LTPrho22},
since $\rho\cket{\gamma_1}=\rho\cket{\gamma_2}=0$.
Note that this is a deterministic CCS, by Proposition~\ref{prop:CCSLTPPCdeterm},
although it may also be illustrative to check this directly, by~\eqref{eq:dxdcondphi}.

% ******************************************************************************
% ******************************************************************************
\section{Conclusions}
\label{sec:concl}
In the paper we revisited the concept of noncommuting common cause, a concept which was introduced into the literature as a generalization of the classical concept of the common cause. Applying noncommuting common causes in causal explanations had plenty of beneficial consequences, among them the rehabilitation of the Common Cause Principle in local physical theories with locally finite degrees of freedom and the existence of a local common causal explanation of the EPR-Bell scenario. 

The concept of noncommuting common causes, however, received sharp criticism. The two most important criticisms were the triviality objection and the lack of explanatory force. In reply to the first objection, in the paper we have shown that the existence of strongly trivial CCSs compromise noncommutative common causal explanations just as little as the existence of classical atomic CCSs compromise classical common causal explanations. Strongly trivial CCSs form only a negligibly small part of the huge set of CCSs. In reply to the second objection, we have shown that although the law of total probability does not a priori hold for noncommuting CCSs, still it can hold for certain noncommuting CCSs. Since the satisfaction of the law of total probability is an indispensable part of a common causal explanation, it can be used as an extra filter in the search for noncommutative CCSs. 

In the paper, we also touched upon a recent debate on the role of the EPR argument in Bell's theorems. We have shown that deterministic common causes cannot be derived from perfect correlations if the underlying event space is not commutative. In other words, perfect correlations can have  noncommuting common causes which are indeterministic. 

Although our paper intends to be a kind of defense of noncommutative common causal explanations, we also admitted some limitations. Perfect correlations can have a noncommutative atomic common causal explanation which satisfies  the law of total probability only in the weak sense, that if if the noncommuting part of the CCS gets zero probability. In the general nonatomic case, however, the law of total probability does not constrain the CCSs to be weakly commutative and this opens up a way to explain quantum correlations by noncommuting common causes. Noncommutative common causes screen off correlations just like commutative ones and if they also satisfy the law of total probability for the given state, they can be interpreted as a finer description of the situation in question. And this is all we require from a neat common causal explanation. 

Unfortunately, it is this latter feature of a finer description which is lacking in the noncommutative common causal explanation given by  \citep{Hofer-Szabo-2012,Hofer-Szabo-2013b}. Namely, if their CCS satisfied the law of total probability, then Bell's inequalities would follow by simply summing up for the common causes. Thus, we think that the furthest one can go in the common causal explanation of the  EPR-Bell scenario is the one given by \citet{Hofer-Szabo-2012,Hofer-Szabo-2013b}, namely, (local) noncommuting common causes violating the law of total probability.

\vspace{.3in}
\noindent
\textbf{Acknowledgments.} This work has been supported by the Hungarian National Research, Development and Innovation Office (grant nos. K-134275, K-134983 and KKP-133827). We thank Péter Vecsernyés for valuable discussions.

\appendix
% ******************************************************************************
% ******************************************************************************
\section{The properties of the two-qubit CCSs}
\label{appsec:2x2}

% ******************************************************************************
\subsection{Summary of the two-qubit examples}
\label{appsec:2x2.summary}

In Table~\ref{tab:examples} we collect the properties of the CCSs of the events $A$ and $B$ \eqref{eq:2x2AB}
in two-qubit systems (see Section~\ref{sec:prob.2x2}).
It is illustrative to recall Propositions~\ref{prop:commatomicCCSwtriv}, \ref{prop:prodvecCCS}, \ref{prop:CommLTP}, \ref{prop:classmaxodet}, \ref{prop:CCSLTPPCdeterm}, \ref{prop:aCCSLTPPCComm} and \ref{prop:aCCSLTPPCComm2}
to see how some of the different columns of the table are related.

\begin{sidewaystable}
\renewcommand{\arraystretch}{1.2}
\centering
\begin{tabular}{cccccccc}
\hline
partition & CCS of \eqref{eq:2x2AB} & atomic & commuting & 
  product & trivial &
  law of total prob. & deterministic \\
\hline
\hline
\eqref{eq:CCStrivAB.2}   & yes & no     & yes &
  (yes) & weakly &
  yes & yes \\
%\hline
\eqref{eq:CCSclass},\eqref{eq:CCStrivAB.4}     & yes & yes    & yes &
  yes & strongly &
  yes & yes \\
%\hline
\eqref{eq:CCSGabor}      & yes & yes    & no  &
  yes & strongly &
  yes & no \\
%\hline
\eqref{eq:CCSclassU}     & yes & yes    & iff $\theta\in\pi\field{Z}$ &
  yes & strongly &
  \tlcell{for $\cket{\psi}$ in \eqref{eq:CCSclassULTPstate},}{Appendix~\ref{appsec:2x2.CCSLTP}} & iff $\theta\in\pi\field{Z}$\\
%\hline
\eqref{eq:CCStwist}      & yes & yes    & iff $\theta\in\pi\field{Z}$ &
  iff $\theta\in\pi\field{Z}$ & \tlcell{weakly,}{str.~iff $\theta\in\pi\field{Z}$}  &
  \tlcell{for $\cket{\psi}$ in \eqref{eq:CCStwistLTPstate},}{Appendix~\ref{appsec:2x2.CCSLTP}} & iff $\theta\in\pi\field{Z}$\\
%\hline
\eqref{eq:CCSBell}       & yes & yes    & no &
  iff $\theta\in\pi\field{Z}$ & \tlcell{weakly,}{str.~iff $\theta\in\pi\field{Z}$}  &
  ? & no \\
%\hline
\eqref{eq:CCShyper}      & yes & yes    & no &
  no & weakly &
  ? & no \\
%\hline
\eqref{eq:CCSntrat}      & for $\cket{\psi}$ in~\eqref{eq:Bell0} & yes    & \tlcell{weakly,}{yes iff $\theta\in\pi\field{Z}$} &
  iff $\theta\in\pi\field{Z}$ & \tlcell{non-,}{str.~iff $\theta\in\pi\field{Z}$} &
  for $\cket{\psi}$ in~\eqref{eq:Bell0} & yes \\
%\hline
\eqref{eq:CCSntratU}     & for $\cket{\psi}$ in~\eqref{eq:Bell0} & yes    & no &
  iff $\theta\in\pi\field{Z}$ & \tlcell{non-,}{str.~iff $\theta\in\pi\field{Z}$} &
  no & no \\
%\hline
\eqref{eq:CCS22ntrat}  & for $\cket{\psi}$ in~\eqref{eq:Bell0} & no & iff $\theta\in\pi\field{Z}$ &
  (iff $\theta\in\pi\field{Z}$) & \tlcell{non-,}{wk.~iff $\theta\in\pi\field{Z}$} &
  for $\cket{\psi}$ in~\eqref{eq:Bell0} & yes \\
%\hline
\eqref{eq:CCS22ntratU} & for $\cket{\psi}$ in~\eqref{eq:Bell0} & no & no &
  (iff $\theta\in\pi\field{Z}$) & non- &
  no & no \\
%\hline
\eqref{eq:CLTP}          & no & yes     & no &
  no & (no) &
  for $\cket{\psi}$ in \eqref{eq:CLTPstate} & (no) \\
\eqref{eq:CCSntratC}      & for $\rho$ in~\eqref{eq:2x2perfcorrABstate.gen} & yes    & \tlcell{weakly,}{yes iff $c\;\text{or}\,s=0$} &
  iff $c\;\text{or}\,s=0$ & \tlcell{non-,}{str.~iff $c\;\text{or}\,s=0$} &
  for $\rho$ in~\eqref{eq:2x2perfcorrABstate.gen} & yes \\
\eqref{eq:CCS22ntratC}  & for $\rho$ in~\eqref{eq:2x2perfcorrABstate.gen} & no & \tlcell{weakly,}{yes iff $c\;\text{or}\,s=0$} &
  (iff $c\;\text{or}\,s=0$) & \tlcell{non-,}{wk.~iff $c\;\text{or}\,s=0$} &
  for $\rho$ in~\eqref{eq:2x2perfcorrABstate.gen} & yes \\
\hline
\end{tabular}
\caption{Summary of the properties of the examples constructed in the main text.
(CCSs in \eqref{eq:CCStrivAB} are understood for two-qubit systems.
(yes)/(no)/(\dots): irrelevant.) 
Note that the CCSs in \eqref{eq:CCSntratC} and \eqref{eq:CCS22ntratC} are given by 
the more general (complex) parametrization of \eqref{eq:CCSntrat} and \eqref{eq:CCS22ntrat}.}
\label{tab:examples}
\end{sidewaystable}

% ******************************************************************************
\subsection{The law of total probability for the CCSs given in (\ref{eq:CCSclassU})}
\label{appsec:2x2.CCSLTP}

We have a one-parameter family of CCSs~\eqref{eq:CCSclassU} of the events $A$ and $B$ \eqref{eq:2x2AB}
and a one-parameter family of state vectors~\eqref{eq:CCSclassULTPstate}.
The question is if this two-parameter freedom is enough to find parameter values,
for which the state and CCS obey the law of total probability~\eqref{eq:2x2LTP}.

The $q_k = \abs{\skalp{\psi}{\gamma_k}}^2$ probabilities of the CCS are 
\begin{subequations}
\begin{align}
q_0 = q_3 &= \frac12\bigabs{a(c^2-s^2)+b2cs}^2=\frac12\bigabs{a\cos(\theta)+b\sin(\theta)}^2,\\
q_1 = q_2 &= \frac12\bigabs{a2cs-b(c^2-s^2)}^2=\frac12\bigabs{a\sin(\theta)-b\cos(\theta)}^2,
\end{align}
\end{subequations}
derived from~\eqref{eq:CCSclassU} and~\eqref{eq:CCSclassULTPstate}.
Using these, the law of total probability~\eqref{eq:2x2LTP} reads as
\begin{equation}
\bigabs{a\cos(\theta)+b\sin(\theta)}^2
\begin{bmatrix}
c^4+s^4\\
2c^2s^2\\
2c^2s^2\\
c^4+s^4
\end{bmatrix}
+\bigabs{a\sin(\theta)-b\cos(\theta)}^2
\begin{bmatrix}
2c^2s^2\\
c^4+s^4\\
c^4+s^4\\
2c^2s^2
\end{bmatrix}=
\begin{bmatrix}
\abs{a}^2\\
\abs{b}^2\\
\abs{b}^2\\
\abs{a}^2
\end{bmatrix},
\end{equation}
also from~\eqref{eq:CCSclassU} and~\eqref{eq:CCSclassULTPstate}.
Using $1^2=(c^2+s^2)^2=c^4+s^4+2c^2s^2$
and $2c^2s^2=\frac12\sin^2(\theta)$ and
$c^4+s^4=1-\frac12\sin^2(\theta)=\frac12(1+\cos^2(\theta))$,
we end up with the equations
\begin{subequations}
\begin{align}
\bigabs{a\cos(\theta)+b\sin(\theta)}^2\frac12(1+\cos^2(\theta)) + \bigabs{a\sin(\theta)-b\cos(\theta)}^2\frac12\sin^2(\theta) = \abs{a}^2,\\
\bigabs{a\cos(\theta)+b\sin(\theta)}^2\frac12\sin^2(\theta) + \bigabs{a\sin(\theta)-b\cos(\theta)}^2\frac12(1+\cos^2(\theta)) = \abs{b}^2.
\end{align}
\end{subequations}
Writing out the products (considering that the parameters are real here),
and collecting the coefficients of $a$ and $b$, we end up with
\begin{subequations}
\label{eq:quadr0}
\begin{align}
\label{eq:quadr0.1}
\begin{split}
& a^2 \Bigl(\cos^2(\theta)\frac12(1+\cos^2(\theta)) +\sin^2(\theta)\frac12\sin^2(\theta) \Bigr)\\
&\qquad+2ab \cos(\theta)\sin(\theta)\Bigl(\frac12(1+\cos^2(\theta))-\frac12\sin^2(\theta) \Bigr)\\
&\qquad+b^2 \Bigl(\sin^2(\theta)\frac12(1+\cos^2(\theta))+\cos^2(\theta)\frac12\sin^2(\theta) \Bigr) = a^2,
\end{split}\\
\label{eq:quadr0.2}
\begin{split}
& a^2 \Bigl(\cos^2(\theta)\frac12\sin^2(\theta)+\sin^2(\theta)\frac12(1+\cos^2(\theta)) \Bigr)\\
&\qquad-2ab \cos(\theta)\sin(\theta)\Bigl(\frac12(1+\cos^2(\theta))-\frac12\sin^2(\theta) \Bigr)\\
&\qquad+b^2 \Bigl(\sin^2(\theta)\frac12\sin^2(\theta)+\cos^2(\theta)\frac12(1+\cos^2(\theta)) \Bigr) = b^2.
\end{split}
\end{align}
\end{subequations}
Recasting these we have
\begin{subequations}
\begin{equation}
\label{eq:quadr}
\begin{bmatrix}a&b\end{bmatrix}
\begin{bmatrix} u&v\\ v&-u \end{bmatrix}
\begin{bmatrix} a\\b \end{bmatrix} =0,\qquad
\begin{bmatrix}a&b\end{bmatrix}
\begin{bmatrix} -u&-v\\ -v&u \end{bmatrix}
\begin{bmatrix} a\\b \end{bmatrix} =0,
\end{equation}
where
\begin{align}
\label{eq:u}
\begin{split}
u &= \cos^2(\theta)\frac12(1+\cos^2(\theta)) +\frac12\sin^4(\theta) -1 \\
  &= - \sin^2(\theta)\frac12(1+\cos^2(\theta))-\frac12\cos^2(\theta)\sin^2(\theta),
\end{split}\\
\label{eq:v}
v &= \cos(\theta)\sin(\theta)\Bigl(\frac12(1+\cos^2(\theta))-\frac12\sin^2(\theta) \Bigr),
\end{align}
\end{subequations}
so the two equations are identical.
It is important to see that
the coefficient of $b^2$ is the opposite of the coefficient of $a^2$
in both~\eqref{eq:quadr0.1} and~\eqref{eq:quadr0.2}, which is the equality in~\eqref{eq:u}.
(It can be seen by straightforward trigonometry.)
This ensures the existence of real solutions for this \emph{homogeneous} quadratic system.
We also simplify $u$ and $v$,
\begin{subequations}
\label{eq:simpluv}
\begin{align}
u &= -\frac12\sin^2(\theta)(1+2\cos^2(\theta)),\\
v &= \cos^3(\theta)\sin(\theta),
\intertext{and we also have}
u^2+v^2 &= \frac14\sin^2(\theta)(1+3\cos^2(\theta)).
\end{align}
\end{subequations}
Now let us have the angle parametrization also for the state vector~\eqref{eq:CCSclassULTPstate} as
$a:=\cos(\xi/2)$, $b:=\sin(\xi/2)$ for $0\leq\xi\leq\pi$.
Then the quadratic equations are of the form
\begin{equation}
\begin{split}
\begin{bmatrix}a&b\end{bmatrix}
\begin{bmatrix} u&v\\ v&-u \end{bmatrix}
\begin{bmatrix} a\\b \end{bmatrix}
&= u \cos^2(\xi/2) + v 2\cos(\xi/2)\sin(\xi/2) - u \sin^2(\xi/2)\\
&= u \cos(\xi) + v \sin(\xi) = 0.
\end{split}
\end{equation}
This expresses the orthogonality of the vectors
$[u,v]$ and $[\cos(\xi),\sin(\xi)]$.
(Expressing $[u,v]/\sqrt{u^2+v^2} = [\cos(\Theta),\sin(\Theta)]$,
we have that the law of total probability holds for the states for which $\xi=\Theta(\theta)\pm \pi/2$;
however, we proceed without the angle $\Theta$.)
To find the $\xi$ parameters of the state, we have
\begin{equation}
\begin{bmatrix}
\cos(\xi)\\
\sin(\xi)
\end{bmatrix}
=\begin{bmatrix}
v/\sqrt{u^2+v^2}\\
-u/\sqrt{u^2+v^2}\\
\end{bmatrix},
\end{equation}
(the opposite would also be suitable, this choice is to fulfill $0\leq\xi\leq\pi$)
then
\begin{equation}
\label{eq:theta}
\begin{bmatrix} a\\ b \end{bmatrix}
=\begin{bmatrix} \cos(\xi/2)\\ \sin(\xi/2) \end{bmatrix}
=\begin{bmatrix}
\sqrt{\frac{1+v/\sqrt{u^2+v^2}}{2}}\\
\sqrt{\frac{1-v/\sqrt{u^2+v^2}}{2}}
\end{bmatrix}
=\begin{bmatrix}
\sqrt{\frac{1+2\cos^3(\theta)/\sqrt{1+3\cos^2(\theta)}}{2}}\\
\sqrt{\frac{1-2\cos^3(\theta)/\sqrt{1+3\cos^2(\theta)}}{2}}
\end{bmatrix},
\end{equation}
by~\eqref{eq:simpluv}.
The function $\theta\mapsto\xi(\theta)$ is continuous, $\xi(0)=0$, $\xi(\pi)=\pi$,
so for all the $\theta$ parameters of the CCS there exists a state of suitable parameter $\xi$
(and vice-versa),
see the plots in Figure~\ref{fig:theta}.

\begin{figure}
    \centering
\includegraphics{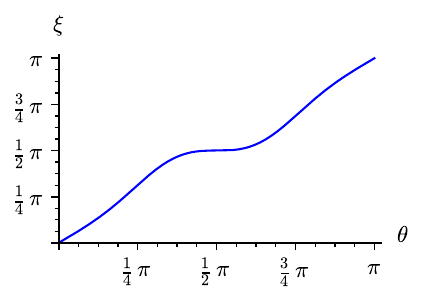}
\includegraphics{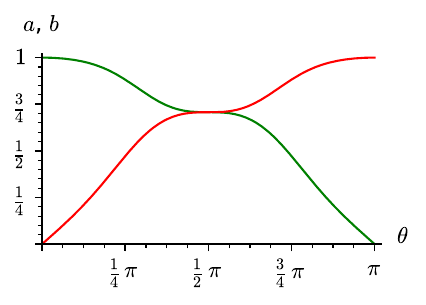} 
    \caption{Plots of the functions given implicitly in~\eqref{eq:theta}.
    (Blue: $\xi(\theta)$, green $a(\theta)=\cos(\xi(\theta)/2)$, red $b(\theta)=\sin(\xi(\theta)/2)$.)}
    \label{fig:theta}
\end{figure}

The CCS given by~\eqref{eq:CCSclassUspec} and the state vector given in~\eqref{eq:CCSclassULTPstatespec} are just these,
evaluated for the value $\theta=\pi/4$.
(The value $\theta=\pi/2$ looks more symmetric, however, $\Delta_\psi(A,B)=0$ then.)
Then $\cos(\theta)=\sin(\theta)=1/\sqrt{2}$,
$c=\cos(\theta/2)=\sqrt{\frac{1+\cos(\theta)}{2}}=\sqrt{\frac{\sqrt{2}+1}{2\sqrt{2}}}$,
$s=\sin(\theta/2)=\sqrt{\frac{1-\cos(\theta)}{2}}=\sqrt{\frac{\sqrt{2}-1}{2\sqrt{2}}}$,
leading to~\eqref{eq:CCSclassUspec}.
We also have $u=-1/2$ and $v=1/4$, 
then $\cos(\xi)=1/\sqrt{5}$ and $\sin(\xi)=2/\sqrt{5}$,
then $a=\cos(\xi/2)=\sqrt{\frac{1+\cos(\xi)}{2}}=\sqrt{\frac{\sqrt{5}+1}{2\sqrt{5}}}$ and
     $b=\sin(\xi/2)=\sqrt{\frac{1-\cos(\xi)}{2}}=\sqrt{\frac{\sqrt{5}-1}{2\sqrt{5}}}$,
so we have the state vector~\eqref{eq:CCSclassULTPstatespec}.

\newpage
% ******************************************************************************
% ******************************************************************************
\bibliography{bibliography.bib}{}

\begin{thebibliography}{22}
\providecommand{\natexlab}[1]{#1}
\providecommand{\url}[1]{\texttt{#1}}
\expandafter\ifx\csname urlstyle\endcsname\relax
  \providecommand{\doi}[1]{doi: #1}\else
  \providecommand{\doi}{doi: \begingroup \urlstyle{rm}\Url}\fi

\bibitem[Bell(1981/2004)]{Bell-2004}
John~Stuart Bell.
\newblock \emph{Bertlmann's socks and the nature of reality}, page 139–158.
\newblock Collected Papers on Quantum Philosophy. Cambridge University Press,
  1981/2004.
\newblock ISBN 9780511815676.
\newblock \doi{10.1017/CBO9780511815676}.
\newblock URL \url{https://doi.org/10.1017/CBO9780511815676}.

\bibitem[Butterfield(1989)]{Butterfield-1989}
Jeremy Butterfield.
\newblock A spacetime approach to the {B}ell inequality.
\newblock In James~T. Cushing and Ernan McMullin, editors, \emph{Philosophical
  Consequences of Quantum Theory}, Studies in Science and the Humanities from
  the Reilly Center for Science, Technology, and Values, pages 114--144. Notre
  Dame Press, 1989.
\newblock ISBN 9780268089696.
\newblock URL
  \url{https://undpress.nd.edu/9780268015794/philosophical-consequences-of-quantum-theory/}.

\bibitem[Cavalcanti and Lal(2014)]{Cavalcanti-2014}
Eric~G Cavalcanti and Raymond Lal.
\newblock On modifications of {R}eichenbach's principle of common cause in
  light of {B}ell's theorem.
\newblock \emph{Journal of Physics A: Mathematical and Theoretical},
  47\penalty0 (42):\penalty0 424018, oct 2014.
\newblock \doi{10.1088/1751-8113/47/42/424018}.
\newblock URL \url{https://dx.doi.org/10.1088/1751-8113/47/42/424018}.

\bibitem[Gömöri and Hoefer(2023)]{Gomori-2023}
Márton Gömöri and Carl Hoefer.
\newblock Classicality and {B}ell's theorem.
\newblock \emph{European Journal for Philosophy of Science}, 13\penalty0 (3),
  sep 2023.
\newblock ISSN 1879-4920.
\newblock \doi{10.1007/s13194-023-00531-y}.
\newblock URL \url{http://dx.doi.org/10.1007/s13194-023-00531-y}.

\bibitem[Haag(1992)]{Haag-1992}
Rudolf Haag.
\newblock \emph{Local Quantum Physics}.
\newblock Springer Berlin Heidelberg, 1992.
\newblock ISBN 9783642973062.
\newblock \doi{10.1007/978-3-642-97306-2}.
\newblock URL \url{http://dx.doi.org/10.1007/978-3-642-97306-2}.

\bibitem[Hofer-Szabó and Vecsernyés(2011)]{Hofer-Szabo-2011}
Gábor Hofer-Szabó and Péter Vecsernyés.
\newblock Reichenbach’s common cause principle in algebraic quantum field
  theory with locally finite degrees of freedom.
\newblock \emph{Foundations of Physics}, 42\penalty0 (2):\penalty0 241–255,
  sep 2011.
\newblock ISSN 1572-9516.
\newblock \doi{10.1007/s10701-011-9594-8}.
\newblock URL \url{http://dx.doi.org/10.1007/s10701-011-9594-8}.

\bibitem[Hofer-Szabó and Vecsernyés(2012)]{Hofer-Szabo-2012}
Gábor Hofer-Szabó and Péter Vecsernyés.
\newblock Noncommuting local common causes for correlations violating the
  {C}lauser–{H}orne inequality.
\newblock \emph{Journal of Mathematical Physics}, 53\penalty0 (12):\penalty0
  122301, nov 2012.
\newblock ISSN 1089-7658.
\newblock \doi{10.1063/1.4763468}.
\newblock URL \url{http://dx.doi.org/10.1063/1.4763468}.

\bibitem[Hofer-Szabó and Vecsernyés(2013{\natexlab{a}})]{Hofer-Szabo-2013a}
Gábor Hofer-Szabó and Péter Vecsernyés.
\newblock Noncommutative common cause principles in algebraic quantum field
  theory.
\newblock \emph{Journal of Mathematical Physics}, 54\penalty0 (4), April
  2013{\natexlab{a}}.
\newblock ISSN 1089-7658.
\newblock \doi{10.1063/1.4801783}.
\newblock URL \url{http://dx.doi.org/10.1063/1.4801783}.

\bibitem[Hofer-Szabó and Vecsernyés(2013{\natexlab{b}})]{Hofer-Szabo-2013b}
Gábor Hofer-Szabó and Péter Vecsernyés.
\newblock Bell inequality and common causal explanation in algebraic quantum
  field theory.
\newblock \emph{Studies in History and Philosophy of Science Part B: Studies in
  History and Philosophy of Modern Physics}, 44\penalty0 (4):\penalty0
  404--416, 2013{\natexlab{b}}.
\newblock ISSN 1355-2198.
\newblock \doi{https://doi.org/10.1016/j.shpsb.2013.08.001}.
\newblock URL
  \url{https://www.sciencedirect.com/science/article/pii/S1355219813000683}.

\bibitem[Hofer-Szabó et~al.(2013)Hofer-Szabó, Rédei, and
  Szabó]{Hofer-Szabo-2013c}
Gábor Hofer-Szabó, Miklós Rédei, and László~E. Szabó.
\newblock \emph{The Principle of the Common Cause}.
\newblock Cambridge University Press, may 2013.
\newblock ISBN 9781139094344.
\newblock \doi{10.1017/cbo9781139094344}.
\newblock URL \url{http://dx.doi.org/10.1017/CBO9781139094344}.

\bibitem[Lazarovici(2014)]{Lazarovici-2014}
Dustin Lazarovici.
\newblock Lost in translation: {A} comment on ``{N}oncommutative causality in
  algebraic quantum field theory''.
\newblock In Maria~Carla Galavotti, Dennis Dieks, Wenceslao~J. Gonzalez,
  Stephan Hartmann, Thomas Uebel, and Marcel Weber, editors, \emph{New
  Directions in the Philosophy of Science}, pages 555--560. Springer
  International Publishing, Cham, 2014.
\newblock ISBN 978-3-319-04382-1.
\newblock \doi{10.1007/978-3-319-04382-1_39}.
\newblock URL \url{https://doi.org/10.1007/978-3-319-04382-1_39}.

\bibitem[Maudlin(2014{\natexlab{a}})]{Maudlin-2014a}
Tim Maudlin.
\newblock What {B}ell did.
\newblock \emph{Journal of Physics A: Mathematical and Theoretical},
  47\penalty0 (42):\penalty0 424010, oct 2014{\natexlab{a}}.
\newblock \doi{10.1088/1751-8113/47/42/424010}.
\newblock URL \url{https://dx.doi.org/10.1088/1751-8113/47/42/424010}.

\bibitem[Maudlin(2014{\natexlab{b}})]{Maudlin_2014b}
Tim~W Maudlin.
\newblock Reply to {C}omment on ‘{W}hat {B}ell did’.
\newblock \emph{Journal of Physics A: Mathematical and Theoretical},
  47\penalty0 (42):\penalty0 424012, oct 2014{\natexlab{b}}.
\newblock \doi{10.1088/1751-8113/47/42/424012}.
\newblock URL \url{https://dx.doi.org/10.1088/1751-8113/47/42/424012}.

\bibitem[Reichenbach(1956)]{Reichenbach-1956}
Hans Reichenbach.
\newblock \emph{The Direction of Time}.
\newblock University of California Press, Berkeley, 1956.

\bibitem[Rédei(1997)]{Redei-1997}
Miklós Rédei.
\newblock Reichenbach’s common cause principle and quantum field theory.
\newblock \emph{Foundations of Physics}, 27\penalty0 (10):\penalty0
  1309–1321, October 1997.
\newblock ISSN 1572-9516.
\newblock \doi{10.1007/bf02551514}.
\newblock URL \url{http://dx.doi.org/10.1007/BF02551514}.

\bibitem[Rédei(1998)]{Redei-1998}
Miklós Rédei.
\newblock \emph{Quantum Logic in Algebraic Approach}.
\newblock Springer Netherlands, 1998.
\newblock ISBN 9789401590266.
\newblock \doi{10.1007/978-94-015-9026-6}.
\newblock URL \url{http://dx.doi.org/10.1007/978-94-015-9026-6}.

\bibitem[Rédei and Summers(2002)]{Redei-2002}
Miklós Rédei and Stephen~J. Summers.
\newblock Local primitive causality and the common cause principle in quantum
  field theory.
\newblock \emph{Foundations of Physics}, 32\penalty0 (3):\penalty0 335–355,
  2002.
\newblock ISSN 0015-9018.
\newblock \doi{10.1023/a:1014869211488}.
\newblock URL \url{http://dx.doi.org/10.1023/A:1014869211488}.

\bibitem[Rédei and Summers(2007{\natexlab{a}})]{Redei-2007a}
Miklós Rédei and Stephen~J. Summers.
\newblock Remarks on causality in relativistic quantum field theory.
\newblock \emph{International Journal of Theoretical Physics}, 46\penalty0
  (8):\penalty0 2053–2062, August 2007{\natexlab{a}}.
\newblock ISSN 1572-9575.
\newblock \doi{10.1007/s10773-006-9299-5}.
\newblock URL \url{http://dx.doi.org/10.1007/s10773-006-9299-5}.

\bibitem[Rédei and Summers(2007{\natexlab{b}})]{Redei-2007b}
Miklós Rédei and Stephen~Jeffrey Summers.
\newblock Quantum probability theory.
\newblock \emph{Studies in History and Philosophy of Science Part B: Studies in
  History and Philosophy of Modern Physics}, 38\penalty0 (2):\penalty0
  390–417, June 2007{\natexlab{b}}.
\newblock ISSN 1355-2198.
\newblock \doi{10.1016/j.shpsb.2006.05.006}.
\newblock URL \url{http://dx.doi.org/10.1016/j.shpsb.2006.05.006}.

\bibitem[Werner(2014{\natexlab{a}})]{Werner-2014a}
Reinhard~F. Werner.
\newblock Comment on `{W}hat {B}ell did'.
\newblock \emph{Journal of Physics A: Mathematical and Theoretical},
  47\penalty0 (42):\penalty0 424011, October 2014{\natexlab{a}}.
\newblock ISSN 1751-8121.
\newblock \doi{10.1088/1751-8113/47/42/424011}.
\newblock URL \url{http://dx.doi.org/10.1088/1751-8113/47/42/424011}.

\bibitem[Werner(2014{\natexlab{b}})]{Werner-2014b}
Reinhard~F. Werner.
\newblock What {M}audlin replied to, 2014{\natexlab{b}}.
\newblock URL \url{https://arxiv.org/abs/1411.2120}.

\bibitem[Wootters(1998)]{Wootters-1998}
William~K. Wootters.
\newblock Entanglement of formation of an arbitrary state of two qubits.
\newblock \emph{Phys. Rev. Lett.}, 80:\penalty0 2245--2248, Mar 1998.
\newblock \doi{10.1103/PhysRevLett.80.2245}.
\newblock URL \url{https://link.aps.org/doi/10.1103/PhysRevLett.80.2245}.

\end{thebibliography}

% ******************************************************************************
% ******************************************************************************
\end{document}